\documentclass[10pt]{scrartcl}
\usepackage[dvips]{color}
\usepackage{epsfig}
\usepackage{subcaption}
\usepackage{etex}

\usepackage{tikz}
\usetikzlibrary{calc, positioning, fit}
\newcommand{\ybot}{-6em}
\newcommand{\ymiddiff}{0.8em}
\newcommand{\xfull}{5.5em}

\newcommand{\onpathdeco}[2]{
\fill (#1) circle [radius=0.14em]; 
\node[xshift = -0.6em, yshift= 0.18em] at (#1){#2};
} 
\tikzstyle{flow} = [->, dashed, gray!120]
\tikzstyle{onpath} = [inner sep=0.15em, fill = white, circle, draw]
\tikzstyle{offpath} = [inner sep=0.11em, draw, fill, circle]
\tikzstyle{mathcoord} = [remember picture, baseline=-0.55ex, inner sep = 0.1ex, every node/.style = {font = \strut}, remember picture, inner sep = -0.3ex, outer sep = 0.1ex]
\tikzstyle{labeloffset} = [xshift = 0em, yshift = .4em, anchor = west]

\usepackage{amsmath,amssymb,dsfont,paralist,floatflt,graphicx}
\usepackage[thmmarks,amsmath,standard]{ntheorem}

\newtheorem{df}{Definition}[section]
\newtheorem{lm}[df]{Lemma}
\newtheorem{ex}[df]{Example}

\newtheorem{theo}[df]{Theorem}
\newtheorem{cor}[df]{Corollary}
\newtheorem{ob}[df]{Observation}

\newcommand{\nat}{\mathbb{N}}

\newcommand{\seml}{[\![}
\newcommand{\semr}{]\!]}

\newcommand{\pimplies}{\stackrel{\mathrm{+}}{\rightarrow}}

\newcommand{\pos}{\mathrm{pos}}
\newcommand{\height}{\mathrm{height}}
\newcommand{\cut}{\mathrm{cut}}

\newcommand{\RMSO}{\mathrm{RMSO}}
\newcommand{\MSO}{\mathrm{MSO}}
\newcommand{\bMSO}{\mathrm{BMSO}}
\newcommand{\bFO}{\mathrm{BFO}}
\newcommand{\bFOmod}{\mathrm{BFO}\! + \! \mathrm{mod}}

\newcommand{\FO}{\mathrm{FO}}
\newcommand{\FOmod}{\mathrm{FO}\! + \! \mathrm{mod}}
\newcommand{\FOtmod}{\mathrm{FO}_t\! + \! \mathrm{mod}}

\newcommand{\Free}{\mathrm{Free}}
\newcommand{\rt}{\mathrm{root}}
\newcommand{\edge}{\mathrm{edge}}
\newcommand{\lab}{\mathrm{label}}

\newcommand{\path}{\mathrm{path}}

\newcommand{\TC}{\mathrm{TC}}

\newcommand{\BTC}{\mathrm{BTC}}
\newcommand{\uBTC}{\mathrm{unf}\mbox{-}\mathrm{BTC}}

\newcommand{\maxrk}{{\mathrm{maxrk}}}

\newcommand{\formpath}{\mathrm{form}\mbox{-}\mathrm{path}}
\newcommand{\formlmp}{\mathrm{form}\mbox{-}\mathrm{lmp}}
\newcommand{\formcut}{\mathrm{form}\mbox{-}\mathrm{cut}}
\newcommand{\onlmp}{\mathrm{on}\mbox{-}\mathrm{lmp}}

\newcommand{\head}{\mathrm{head}}
\newcommand{\nextbase}{\mathrm{next}\mbox{-}\mathrm{base}}
\newcommand{\sibl}{\mathrm{sibl}}
\newcommand{\dec}{\mathrm{dec}}

\newcommand{\chk}{\mathrm{check}}

\newcommand{\sel}{\mathrm{sel}}
\newcommand{\listpos}{\mathrm{list}\mbox{-}\mathrm{pos}}

\DeclareMathOperator{\rk}{rk}

\DeclareMathOperator{\size}{size}

\DeclareMathOperator{\n}{-}

\let\S=\Sigma

\def\-{$-$}
\def\|{\hspace{1mm} | \hspace{1mm}}

\def\seq#1#2#3{#1_{#2},\ldots,#1_{#3}} 

\let\u=\cup
\def\g0{\geq0}             

\def\xe{x_\varepsilon}

\def\ui#1{^{(#1)}}

\begin{document}

\title{Characterizing\\ Weighted MSO for Trees by\\ Branching Transitive Closure Logics}

\author{Zolt\'an F\"ul\"op$^a$\thanks{Research of this author was financially supported by the programs
    T\'AMOP-4.2.1/B-09/1/KONV-2010-0005 and T\'AMOP-4.2.2.A-11/1/KONV-2012-0073 of the Hungarian National
    Development Agency.} \, and Heiko Vogler$^b$\\
  {\small $^a$ Department of Foundations of Computer Science,
    University of Szeged} \\[-.5ex]
  {\small \'Arp\'ad t\'er 2., H-6720 Szeged, Hungary.
    {fulop@inf.u-szeged.hu}} \smallskip \\
  {\small $^b$ Faculty of Computer Science, Technische Universit\"at
    Dresden} \\[-.5ex]
  {\small Mommsenstr.~13, D-01062 Dresden, Germany.
    {Heiko.Vogler@tu-dresden.de}}}

\date{\today}

\maketitle

\sloppy

\begin{quote}{\bf Abstract:} We introduce the branching transitive
  closure operator on progressing weighted monadic second-order logic formulas where the
  branching corresponds in a natural way to the branching inherent in
  trees. For arbitrary commutative semirings, we prove that
  weighted monadic second order logics on trees  is equivalent to the
  definability by formulas which start with one of the following
  operators: (i) a branching transitive closure or (ii) one existential second-order quantifier followed by one universal first-order quantifier; in both cases the operator is applied to step-formulas over (a) Boolean first-order logic enriched by modulo counting or (b) Boolean monadic-second order logic.   
\end{quote}

{\bf ACM classification:} F.1.1, F.4.1, F.4.3.

{\bf Key words and phrases:} weighted tree automata, weighted monadic second-order logic, transitive closure.

\section{Introduction}\label{introduction}

In \cite{barmak92} monadic second order logic (MSO) for strings is
characterized by the extension of first-order (FO) logic with
unary transitive closure ($\FO+\TC^{[1]}$). In \cite[Thm.10]{bolgasmonzei10}
weighted restricted MSO for strings is characterized by the
application of a (progressing) unary transitive closure operator to
step formulas over FO formulas extended by modulo counting. For trees such
a characterization of MSO in terms of transitive closure existed neither for the weighted nor
for the unweighted case.  In  \cite{tenseg08}
it was proved that MSO on trees is strictly more powerful than
$\FO+\TC^{[1]}$. Moreover, MSO is strictly less powerful than
$\bigcup_{k\geq 1}\FO+\TC^{[k]}$ where $\TC^{[k]}$ denotes the transitive
closure of some binary relation over the set of $k$-tuples of
positions in trees \cite{tiekep09}; even
$\FO+\TC^{[2]}$ contains a tree language which is not definable in MSO
(cf.  \cite[Prop. 4]{tiekep09}). 
This raises the following question: is there a version of transitive
closure for trees which characterizes MSO in the unweighted and the
weighted case?

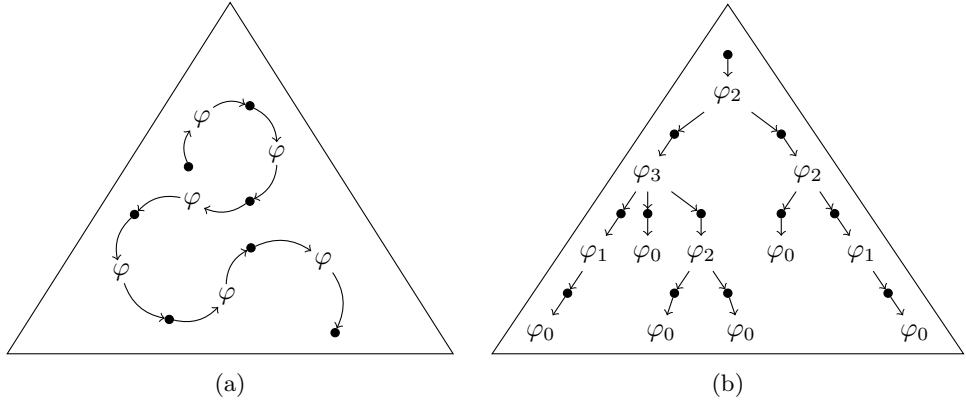
\begin{figure} 
\centering
	\begin{subfigure}[b]{0.49\textwidth}
	\centering
	\begin{tikzpicture}
		\tikzstyle{op} = [offpath]
		\tikzstyle{a1} = [bend left, shorten > = .5em, shorten < = .11em, ->]
		\tikzstyle{a2} = [bend left, shorten < = .5em, shorten > = .11em, ->]
		\tikzstyle{w} = [inner sep = .1em]
		\newcommand{\smalldiff}{.5em}
		\newcommand{\bigdiff}{1.8em}
		\node[op] at (0,0) (n0) {};
		\node[w, shift ={(\smalldiff, \bigdiff)}] at (n0) (n1) {$\varphi$};
		\node[op, shift ={(\bigdiff, \smalldiff)}] at (n1) (n2) {};
		\node[w, shift ={(2* \smalldiff, -\bigdiff)}] at (n2) (n3) {$\varphi$};
		\node[op, shift ={(-2 * \smalldiff, -\bigdiff)}] at (n3) (n4) {};
		\node[w, shift ={(-1.2 * \bigdiff, 0)}] at (n4) (n5) {$\varphi$};
		\node[op, shift ={(-\bigdiff * 1.2, -\smalldiff)}] at (n5) (n6) {};
		\node[w, shift ={(-\smalldiff, -\bigdiff * 1.2)}] at (n6) (n7) {$\varphi$};
		\node[op, shift ={(\bigdiff, -\bigdiff)}] at (n7) (n8) {};
		\node[w, shift ={(\bigdiff * 1.2, \bigdiff /2)}] at (n8) (n9) {$\varphi$};
		\node[op, shift ={(\bigdiff / 2, \bigdiff)}] at (n9) (n10) {};
		\node[w, shift ={(\bigdiff * 1.5, -\smalldiff)}] at (n10) (n11) {$\varphi$};
		\node[op, shift ={(\bigdiff/4, -1.5* \bigdiff)}] at (n11) (n12) {};	

		\foreach \n in {0,1,2,3,4,9, 10,11}{\pgfmathtruncatemacro{\next}{\n + 1} \draw[bend left, ->] (n\n) to (n\next);}
		\foreach \n in {5,6,7,8}{\pgfmathtruncatemacro{\next}{\n + 1} \draw[bend right, ->] (n\n) to (n\next);}
		
		\coordinate[left = 1.4em of n7] (ll);
		\coordinate[right = 2em of n11] (rr);
		\node[fit = (n1) (n2) (n3) (n4) (n5) (n6) (n7) (n8) (n9) (n10) (n11) (ll) (rr), inner sep = 2.4em, yshift = 1.3em] (bound) {};
		\draw (bound.south west) -- (bound.north) -- (bound.south east) -- cycle;
	\end{tikzpicture}
	\caption{}
	\end{subfigure}
	\begin{subfigure}[b]{0.49\textwidth}
	\centering
	\begin{tikzpicture}[level distance = 1.5em, sibling distance = 6em, level 4/.style={sibling distance = 2em}, edge from parent/.style={draw,->}]
		\tikzstyle{op} = [offpath]
		\node[op] (n0) {}
			child {node {$\varphi_2$}
				child {node[op, xshift = 1em] {}
					child {node[xshift = -1em] {$\varphi_3$}
						child {node[op, xshift = 1em] {}
							child {node[xshift = -1em] {$\varphi_1$}
								child {node[op, xshift = -1em] {}
									child {node[xshift = -1em] (n1) {$\varphi_0$}}
								}
							}
						}
						child {node[op] {}
							child {node {$\varphi_0$}}
						}
						child {node[op] {}
							child {node {$\varphi_2$}
								child {node[op] {}
									child {node[xshift = -0.5em] {$\varphi_0$}}
								}
								child {node[op] {}
									child {node[xshift = 0.5em] {$\varphi_0$}}
								}
							}
						}
					}
				}
				child {node[op, xshift = -1em] {}
					child {node[xshift = 1em] {$\varphi_2$}
						child {node[op] {}
							child {node {$\varphi_0$}}
						}
						child {node[op] {}
							child {node[xshift = 1em] {$\varphi_1$}
								child {node[op, xshift = 1em] {}
									child {node[xshift = 1em] (n2) {$\varphi_0$}}
								}
							}
						}
					}
				}
			};
		\node [fit = (n0) (n1) (n2), inner sep = 0.9em, yshift = .8em] (bound) {};
		\draw (bound.south west) -- (bound.north) -- (bound.south east) -- cycle;
	\end{tikzpicture}
	\caption{}
	\end{subfigure} 

\caption{\label{fig:TC}(a) A sequence of positions chosen by
  $\TC^{[1]}(\varphi)$ and
  (b) a ``tree'' of positions chosen by $\BTC(\Phi)$. }
\end{figure}

In this paper we define the concept of {\em  branching
  transitive closure} for trees ($\BTC$) and we characterize MSO on
trees by $\BTC$ applied to FO extended by modulo counting. 
Let us informally explain the concept of $\BTC$ by first recalling how
$\TC^{[1]}$ works. The operator  $\TC^{[1]}$  is applied to a formula $\varphi(x,y)$ with two free
variables $x$ and $y$, called input and output variable, respectively. Then
$\TC^{[1]}(\varphi)$ is interpreted as the transitive closure of the binary
relation induced by $\varphi$. If $\TC^{[1]}(\varphi)$ is interpreted on a
tree, then a sequence of positions is chosen;
these positions might be thought of as intermediate points of a tree-walk
(cf. Fig. \ref{fig:TC}(a)).  Contrary to  $\TC^{[1]}$, the operator $\BTC$ is applied to a finite family $\Phi = (\varphi_k(x,y_1,\ldots,y_k) \mid 0\le k
\le m)$ of formulas where $\varphi_k(x,y_1,\ldots,y_k)$ has the free
input variable $x$ and the free output variables 
$y_1,\ldots,y_k$. Then $\BTC(\Phi)$ is interpreted on a
tree as follows (cf. Fig. \ref{fig:TC}(b)). The interpretation starts by choosing an arbitrary position $v$ as assignment for $x$.
Then, the operator chooses a number $l$ of positions which it will visit (in Fig. \ref{fig:TC}(b) $l=12$); $v$ is one of these positions. Next the operator chooses a branching degree $k$ and thereby the
formula $\varphi_k(x,y_1,\ldots,y_k)$. To each $y_i$ it assigns
a position $v_i$;  in
particular, the $v_i$ are in a certain sense ``below'' $v$ such that
the application of $\varphi_k$ ensures a progress down the tree. Before calling itself on $v_i$, the operator chooses a distribution of $l-1$ to its $k$ recursive calls, i.e., it chooses a sequence $l_1,\ldots,l_k$ of numbers such that $l-1 = l_1 + \ldots + l_k$ and $l_i$ is the number of all positions which are visited in the $i$th recursive call (including $v_i$). Thereafter the operator splits into $k$ copies and, for each $1 \le i \le k$, one copy visits position $v_i$.
This process is iterated where the output positions of an iteration step
become the input positions for the next step. Finally, the formula $\varphi_0(x)$ has to be chosen which finishes the iteration.
Figure 1(b) shows a protocol of this interpretation of $\BTC(\Phi)$ which we will call {\em unfolding}. 
Hence, $\BTC(\Phi)$ reflects in a natural way the branching structure of trees. 
The tree satisfies this unfolding if $\varphi_k(v,v_1,...,v_k)$ holds for every
chosen formula and assigment belonging to the formula.  Moreover, the tree satisfies $\BTC(\Phi)$ if
it satisfies at least one of its unfoldings. For a class ${\cal   L}$ of formulas, we denote by $\BTC({\cal L})$ the class of all formulas of the form $\BTC(\Phi)$ and $\Phi$ is a family of progressing formulas in ${\cal L}$.

In this paper we characterize  MSO by $\BTC(\FOmod)$
where $\FOmod$ is a class of FO formulas extended by modulo counting
(similar to \cite{bolgasmonzei10}). Let us explain how an arbitrary MSO formula is transformed into formulas of
$\BTC(\FOmod)$.  For this we represent the MSO formula by a
finite state tree automaton ${\cal A}$  \cite{thawri68,don70}. Then we use the idea of \cite{tho82} of splitting the given input tree into slices, where
the number $n$ of states completely determines the shape and the number of the
slices. However,  due to the branching inherent in trees, the appropriate definition of a slice was a technical challenge (cf. Section \ref{sec:slices}).  For instance, in Fig.~\ref{fig:slices-phi},  for $n=3$ the input tree $\xi$ is
splitted into the shown slices $\zeta_1,\ldots,\zeta_6$.

The state behaviour of ${\cal A}$ on  $\xi$ induces a
state behaviour on the slices of $\xi$. 
Then, due to the idea invented in \cite{tho82}, the state in which the evaluation of a slice starts can be represented by a position of this slice.  The retrieval of the state
from a position uses the  modulo counting technique. The state behaviour on the
slices is handled by assigning the representing positions to the free variables
of the instances of the $\varphi_k$-formulas.

To understand the idea, let us consider Fig.~\ref{fig:slices-phi}.
Let us assume that ${\cal A}$ has the state set $\{0, 1, 2\}$ and that the  evaluation of the slices $\zeta_1,\ldots,\zeta_6$ are started in
states 0, 1, 1, 2, 0, and 1, respectively.  Then, e.g., the state $0$ in
the slice $\zeta_1$ is represented by position $\varepsilon$, the state $1$ in
$\zeta_2$ by $2111$, and the state $2$
in $\zeta_4$ by $32111$. Hence, the state behaviour on
$\zeta_1$ is handled by $\varphi_2(x,y_1,y_2)$ under the assignment $x
\mapsto \varepsilon, y_1 \mapsto 2111, y_2 \mapsto 32111$.

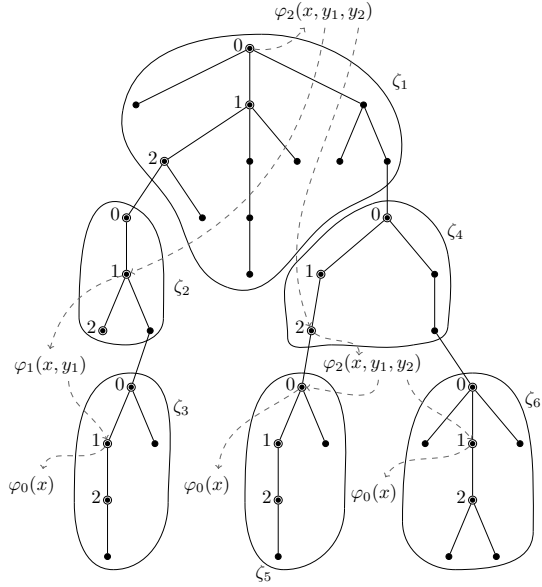
\begin{figure}
\begin{center}
\scalebox{0.71}{
\begin{tikzpicture}[level distance = 3em]
\begin{scope}[every node/.style={offpath}, sibling distance = 6em, level 2/.style={sibling distance = 2.5em}]
	\node[onpath] (z10) {}
		child {node (z13) {}}
		child {node[onpath] (z11) {}
			child { node[onpath, xshift = -2em] (z12) {}
				child {node[xshift = 2em] {}}
			}
			child { node {}
				child {node {}
					child {node (z14) {}}
				}
			}
			child { node {}}
		}
		child {node {}
			child {node {}}
			child {node (z15){}	}
		}
	;
\end{scope}
\foreach \z in {0,1,2}{
	\onpathdeco{z1\z}{\z}
}
\node[fit = (z10) (z11) (z12) (z13) (z14) (z15), inner sep = .5em] (b1) {};
\path (b1.north) edge[out = 180, in = 80] (b1.160) 
		  (b1.160) edge[out = 260, in = 110] ($(b1.180) !.5! (b1.260)$)
		  ($(b1.180) !.5! (b1.260)$) edge[out = 290, in = 160] (b1.255) 
		  (b1.255) edge[out= 340, in = 215] ($(b1.357) !.4! (b1.258)$) 
		  ($(b1.357) !.4! (b1.258)$) edge[out = 35, in = 260] (b1.east)
		  (b1.east) edge[out = 80, in = 0] (b1.north);
\node at (b1.30) {$\zeta_1$};

\begin{scope}[every node/.style={offpath}, sibling distance = 6em, level 2/.style={sibling distance = 2.5em}]
	\node[onpath, xshift = -2em, yshift = -3em] (z20) at (z12) {}
		child { node[onpath] (z21) {}
			child { node[onpath] (z22){}}
			child { node (z23) {}}
	}
	;
	\draw (z12) to (z20);
\end{scope}
 	\foreach \z in {0,1,2}{
 		\onpathdeco{z2\z}{\z}
 	}
\node[fit = (z20) (z21) (z22) (z23), inner sep = .5em] (b2) {};
\path (b2.north) edge[out = 160, in = 90] ($(b2.225) + (-.5em,-0.3em)$) 
		($(b2.225) + (-.5em,-0.3em)$)  edge[out = 270, in = 180] (b2.270) 
		(b2.270) edge[out = 0, in = 270] ($(b2.310)$)
		($(b2.310)$) edge[out = 90, in = 340] (b2.north);
\node[xshift = 1em] at (b2.340) {$\zeta_2$};

\begin{scope}[every node/.style={offpath}, sibling distance = 2.5em]
	\node[onpath, yshift = -3em, xshift = -1em] (z30) at (z23) {}
		child { node[onpath] (z31) {}
			child { node[onpath] (z32){}
				child {node (z33) {}}
			}
		}
		child { node (z34) {}}
	;
\end{scope}
 	\draw (z23) to (z30);
 	\foreach \z in {0,1,2}{
 		\onpathdeco{z3\z}{\z}
 	}
\node[fit = (z30) (z31) (z32) (z33) (z34), inner sep = .5em] (b3) {};
\path (b3.north) edge[out = 180, in = 90] ($(b3.220)+(-1em,0)$)
		  ($(b3.220) + (-1em,0)$) edge[out = 270, in = 180] (b3.260)
		  (b3.260) edge[out = 0, in = 270] (b3.320)
		  (b3.320) edge[out = 90, in = 0] (b3.north);
\node[xshift = .7em] at (b3.60) {$\zeta_3$};

\begin{scope}[every node/.style={offpath}, sibling distance = 5em]
	\node[onpath, yshift = -3em] (z40) at (z15) {}
		child { node[onpath, xshift = -1em] (z41) {}
			child { node[onpath, xshift = -.5em] (z42){}}
		}
		child { node {}
			child {node (z43) {}}
		}
	;
\end{scope}
 	\draw (z15) to (z40);
 	\foreach \z in {0,1,2}{
 		\onpathdeco{z4\z}{\z}
 	}

\node[fit = (z40) (z41) (z42) (z43), inner sep = .5em] (b4) {};
\path (b4.north) edge[out = 200, in = 45] ($(b4.140) !.8! (b4.west)$)
		($(b4.140) !.8! (b4.west)$) edge[out = 225, in = 90] ($(b4.220)+(-0.5em,0)$)
		  ($(b4.220) + (-0.5em,0)$) edge[out = 270, in = 180] (b4.240)
			(b4.240) edge [out = 0, in = 190] (b4.300)
		  (b4.300) edge[out = 10, in = 265] (b4.325)
		  (b4.325) edge[out = 85, in = 280] (b4.15)
		  (b4.15) edge[out = 100, in = 20] (b4.north);
\node[xshift = .4em] at (b4.30) {$\zeta_4$};

\begin{scope}[every node/.style={offpath}, sibling distance = 2.5em]
	\node[onpath, yshift = -3em, xshift = -0.5em] (z50) at (z42) {}
		child { node[onpath] (z51) {}
			child { node[onpath] (z52){}
				child {node (z53) {}}
			}
		}
		child { node (z54) {}}
	;
\end{scope}
 	\draw (z42) to (z50);
 	\foreach \z in {0,1,2}{
 		\onpathdeco{z5\z}{\z}
 	}

\node[fit = (z50) (z51) (z52) (z53) (z54), inner sep = .5em] (b5) {};
\path (b5.north) edge[out = 180, in = 90] ($(b5.220)+(-1em,0)$)
		  ($(b5.220) + (-1em,0)$) edge[out = 270, in = 180] (b5.260)
		  (b5.260) edge[out = 0, in = 270] ($(b5.315) + (.3em,0)$)
		  ($(b5.315) + (.3em,0)$)edge[out = 90, in = 0] (b5.north);
\node[yshift = -2em] at (b5.240) {$\zeta_5$};

\begin{scope}[every node/.style={offpath}, sibling distance = 2.5em]
	\node[onpath, yshift = -3em, xshift = 2em] (z60) at (z43) {}
		child { node (z63) {}}
		child { node[onpath] (z61) {}
			child { node[onpath] (z62){}
				child {node (z64) {}}
				child {node (z65) {}}
			}
		}
		child { node (z66) {}}
	;
\end{scope}
 	\draw (z43) to (z60);
 	\foreach \z in {0,1,2}{
 		\onpathdeco{z6\z}{\z}
 	}

\node[fit = (z60) (z61) (z62) (z63) (z64) (z65) (z66), inner sep = .5em] (b6) {};
\path (b6.north) edge[out = 180, in = 90] ($(b6.220)+(-0.5em,0)$)
		  ($(b6.220) + (-0.5em,0)$) edge[out = 270, in = 180] (b6.260)
		  (b6.260) edge[out = 0, in = 270] (b6.315)
		  (b6.315) edge[out = 90, in = 0] (b6.north);
\node[xshift = 0em] at (b6.50) {$\zeta_6$};

\node[yshift = 1.2em] (p1) at (b1.64) {$\varphi_2(x,y_1,y_2)$};
\node[yshift = -1.4em,xshift = -2em] (p2) at (b2.240){$\varphi_1(x, y_1)$};
\node[xshift = -3.25em] (p3) at (b3.200) {$\varphi_0(x)$};
\node[xshift = -3em] (p5) at (b5.200) {$\varphi_0(x)$};
\node[yshift = -1em] (p4) at (b4.south) {$\varphi_2(x,y_1,y_2)$};
\node[xshift = -2em] (p6) at (b6.200) {$\varphi_0(x)$};

\draw[flow,bend right] (z10) to (p1.220);
\draw[flow,bend left] (p1.274) to (z21);
\draw[flow,bend right] (p1.335) to[out = 360, in = 210] (z42);

\draw[flow,bend right] (z21.200) to (p2.92);
\draw[flow, bend right] (p2.318) to[out = 330, in = 170] (z31.100);
\draw[flow, bend right] (z31) to[out = 30, in = 240] (p3.55);
\draw[flow, bend right] (z50) to[out = 350, in = 220] (p5.55);
\draw[flow, bend right] (z61) to[out = 20, in = 220] (p6.55);
\draw[flow, bend right] (z42) to[out = 350, in = 110] ($(p4.130) + (0,-.4em)$);

\draw[flow, bend right] (p4.290) to[out = 80, in = 170] (z50);
\draw[flow, bend right] (p4.335) to[out = 330, in = 170] (z61.100);
\end{tikzpicture}
}
\end{center}
\caption{\label{fig:slices-phi} An example of slices, representation of states by positions, and
formulas which handle the state behaviour of $\cal A$.}
\end{figure}

For the other inclusion $\BTC(\FOmod) \subseteq \MSO$ we 
simulate every formula  $\BTC(\Phi)$ by a particular $\MSO$-formula of the form
$\exists X. \forall x. \theta(X,x)$ where $\theta(X,x) \in {\cal L}$ and $\exists X$ is the only  second-order quantification (as it
was done in  \cite{bolgasmonzei10} for strings).

In fact, we prove our characterization result in a more general
setting, viz. for weighted MSO logics over semirings \cite{drogas05,drogas07,drogas09}. There, the expression $\xi \models
\varphi$ does not have a Boolean value, indicating whether $\xi$
is a model of $\varphi$ or not; rather, this expression takes a value
in some given semiring \cite{gol99}. The progress down the tree guarantees that no infinite summations occur in the definition of the
semantics of the operator $\BTC$. If the Boolean semiring (with
disjunction and conjunction) is employed, then the classical, unweighted case is
reobtained.  If the semiring of natural numbers is employed, then $\xi \models \varphi$ can be understood as the number of proofs for the claim that $\xi$ is a model of $\varphi$ assuming that atomic formulas have the weight 1 (cf. Example II of \cite{drogas09} where the reader can find more motivating examples). The investigation of many-valued logics has a long tradition. Already ~{\L}ukasiewicz \cite{luk20} and Post \cite{pos21} investigated logics with different degrees of certainty; Birkhoff and von Neumann \cite{birneu36} introduced quantum logics with values in orthomodular lattices as logics of quantum mechanics. In the spirit of \cite{drogas09}, MSO over arbitrary bounded lattices were considered in \cite{drovog12}.

For the proof of our main results, we represent a weighted MSO formula by a weighted tree automaton. This is possible because weighted MSO logics over (commutative) semirings is equivalent
to weighted tree automata \cite{drovog06}. Weighted string automata and weighted tree automata have a rich theory \cite{eil74,salsoi78,wec78,berreu82,kuisal86,sak09,drokuivog09} and they are applied in different areas, like modelling and analysis of weighted distributed systems \cite{fickusmei09}, digital image compression \cite{albkar09}, and natural language processing \cite{moh09,knimay09}.

In our main result (Theorem \ref{main}) we generalize 
\cite[Thm. 10]{bolgasmonzei10} to recognizable weighted tree languages, but
only for commutative semirings. We prove  that the expressive powers of recognizability and of the  logics
\begin{center}
(i)~weighted RMSO\; (ii)~$\BTC({\cal L}_{\mathrm{step}})$, (iii)~~$\exists\forall({\cal L}_{\mathrm{step}})$,
\end{center}
are equivalent, where 
\begin{compactitem}
\item $\RMSO$ stands for restricted $\MSO$,
\item ${\cal L}_{\mathrm{step}}$ is the logic which contains all ${\cal L}$-step formulas (cf.  \cite[Equ. (1)]{bolgasmonzei10}) where ${\cal L} \in \{\bFOmod, \bMSO\}$ and $\bFO$ and $\bMSO$ are the Boolean fragments of FO and MSO, resp., and mod allows modulo counting,  and
\item $\exists\forall({\cal L}_{\mathrm{step}})$ is the logic which contains all formulas of the form $\exists X. \forall y. \varphi$ where  $\varphi~\in~{\cal L}_{\mathrm{step}}$.
\end{compactitem}

The handling of the weights is done in the same way as in
\cite{bolgasmonzei10}, from which we borrow several notations and notions; also we follow their lines of argumentation.
 However, the switch from strings to trees created two technical
 difficulties: (1) the appropriate splitting of an input tree into
 slices and (2) the unique representation of states in a slice in
 order to avoid counting a state behaviour too often. We employed the
 $\bFO$-formulas  $\formcut$ and $\onlmp$, respectively, for handling
 these difficulties (cf. Section \ref{sec:construction-of-Phi}). Moreover, we use \cite[Prop. 18]{mal06a}
 (cf. Lemma \ref{lm:decomp}) for the fact that  the state-value behaviour of a weighted tree
 automaton on an input tree $\xi$ induces a
 state-value behaviour on the slices of $\xi$. This needs the
 commutativity of the semiring multiplication.

In Section 2 we recall general notations on trees, the definitions of weighted tree automata and (fragments of) weighted MSO. In Section 3 we introduce our branching transitive closure operator and illustrate it by means of an example. Section 4 shows the main result of this paper (cf. Theorem \ref{main}); its proof uses results which are proved in Sections  5, 6, and 7. We conclude in Section 8 by indicating some open problems.

We will use a number of macros in weighted MSO, and we introduce them at the places where they are needed first time. For the convenience of the reader we have collected all the macros in an appendix.  

We have tried to make the paper self-contained concerning the formal definitions. This implies that the preliminaries contain the list of all the logics used in this paper. The experienced reader can skip this subsection upon first reading and consult if necessary. 

\section{Preliminaries}
\label{prel}

\subsection{General Notation}

Let $\mathbb{N}$ denote the set $\{0,1,2,\ldots\}$ of natural numbers; let $\mathbb{N}_+ = \mathbb{N} \setminus \{0\}$.
The cardinality of a set $A$ is denoted by $|A|$. Frequently we abbreviate a tuple $(a_1,\ldots, a_k)\in A^k$ by $a_1\ldots a_k$.
Note that $A^0=\{(\,)\}$, and we abbreviate $(\,)$ by $\varepsilon$.

An alphabet is a non-empty and finite set $\Delta$. The set of 
strings over $\Delta$ is denoted by $\Delta^*$. The empty string is denoted by $\varepsilon$ and the length of $w \in \Delta^*$ is denoted by $|w|$.

\subsection{Trees}

We use the usual notions and notations concerning trees, cf., e.g.,
\cite{fulvog09}. For a ranked alphabet $\Sigma$ we denote by
$\Sigma^{(k)}$ the set of symbols of $\Sigma$ having rank $k$ and by
$\mathrm{maxrk}(\Sigma)$ the maximal rank of symbols in $\Sigma$. The
set of $\Sigma$-trees indexed by some set $A$ is denoted by
$T_\Sigma(A)$. In case $A=\emptyset$ we write $T_\S$ for
$T_\Sigma(A)$. Given a tree $\xi \in T_\Sigma(A)$, we denote the set
of its positions by $\pos(\xi) \subseteq \nat_+^*$ (using the usual
Gorn-notation). The usual prefix ordering on $\pos(\xi)$ is denoted by $\leq$.
 We abbreviate a sequence $v_1,\ldots,v_k$ of positions
by $v_{1,k}$. 
For every $\xi \in T_\Sigma(A)$ and $w \in \pos(\xi)$, we denote the
label of $\xi$ at $w$ by $\xi(w)$ and the subtree of $\xi$ at $w$ by
$\xi|_w$. For any set $B \subseteq A \cup \Sigma$, we denote by
$\pos_B(\xi)$ the set $\{ w\in \pos(\xi)\mid \xi(w) \in B\}$. If,
additionally, $\zeta \in T_\Sigma(A)$, then $\xi[\zeta]_w$ denotes the tree obtained from $\xi$ by replacing the subtree at $w$ by $\zeta$.

If $A$ is finite, then we can define the ranked alphabet $(\Sigma \cup A,\rk)$ by $\rk(a) = 0$ for every $a \in A$ and $\rk(\sigma) = \rk_\Sigma(\sigma)$ for every $\sigma \in \Sigma$.

We define the \emph{height}~$\height(\xi)$ and the \emph{size}~$\size(\xi)$  of a tree $\xi \in T_\Sigma(A)$ recursively as follows.  For
every $\alpha \in A \cup \Sigma^{(0)}$, let $\height(\alpha) = 0$ and $\size(\alpha) = 1$, and for
every $\sigma \in \Sigma^{(k)}$ with $k \ge 1$ and $\xi_1,\ldots,\xi_k \in T_\Sigma$, let $\height(\sigma(\xi_1,\ldots,\xi_k)) = 1 + \max\{\height(\xi_i) \mid 1 \leq i \leq k \}$  and $\size(\sigma(\xi_1,\ldots,\xi_k)) = 1 + \sum_{i=1}^k \size(\xi_i)$. In fact, $\size(\xi) = |\pos(\xi)|$ for every $\xi \in T_\Sigma$.

A ranked alphabet $\Sigma$ is {\em monadic} if $\Sigma = \Sigma^{(1)} \cup \Sigma^{(0)}$ and $\Sigma^{(0)}~=~\{e\}$ is a singleton. For such a $\Sigma$  there is an obvious  bijection from $T_\Sigma$ to $(\Sigma^{(1)})^*$ which transforms monadic trees into strings.

\subsection{Weighted Tree Languages and Weighted Tree Automata}

A commutative semiring is an algebra $(S,+,\cdot,0,1)$ where $(S,+,0)$ and $(S,\cdot,1)$ are commutative monoids, $\cdot$ distributes over $+$, and $0$ is absorbing with respect to $\cdot$. As usual, we abbreviate $(S,+,\cdot,0,1)$ by $S$.
\begin{quote}\it In this paper, $S$ will always denote an arbitrary commutative semiring. 
\end{quote}
For more details on semirings we refer the reader to \cite{hebwei98,gol99}. A {\em weighted tree language} is a mapping $r: T_\Sigma \rightarrow S$ for
some ranked alphabet $\Sigma$. In particular, for
every tree language $L \subseteq T_\Sigma$, we denote by
$\mathds{1}_L$ the weighted tree language $\mathds{1}_L: T_\Sigma
\rightarrow S$ with $\mathds{1}_L(\xi) = 1$ for every $\xi \in L$,
and $0$ otherwise. We call $\mathds{1}_L$ the {\em characteristic
  weighted tree language of $L$}. A {\em recognizable step function} \cite{drogas05,drogas07,drogas09}
is a weighted tree language $r: T_\Sigma \rightarrow S$ such that
there are $n \ge 0$, recognizable tree languages $L_1,\ldots,L_n$ over
$\Sigma$ \cite{gecste84,gecste97}, and coefficients $a_1,\ldots,a_n$ in $S$ such that $r = \sum_{i=1}^n a_i \cdot \mathds{1}_{L_i}$.

We recall the concepts of weighted tree automata from \cite{fulvog09}. 
A \emph{weighted tree automaton
  over~$S$} (wta) is a tuple ${\cal A} = (Q, \Sigma, \delta,F)$ where
 $Q$~is a finite, nonempty  set (of \emph{states}),
 $\Sigma$~is a ranked alphabet (of \emph{input symbols}),
$F \subseteq Q$~is a set of \emph{final states}, and
 $\delta$~is a family~$(\delta_k \mid k \in \nat)$ of \emph{weighted
    transitions} with $\delta_k \colon Q^k \times \Sigma^{(k)} \times Q
  \to S$ for every $k \in \nat$.

Let $k \in \nat$. We define $Z_k = \{z_1,\ldots,z_k\}$ (hence $Z_0 = \emptyset$) and,
for every $q_1,\ldots, q_k \in Q$, we define the mapping $h^{q_1 \ldots q_k} \colon T_\Sigma(Z_k) \to S^Q$ recursively as follows.  For every $q \in Q$ and
\begin{itemize}
\item for every $z_i \in Z_k$ we have
$h^{q_1 \ldots q_k}(z_i)_q = 1$ if $q = q_i$, and $0$ otherwise, and
\item for every $\sigma \in \Sigma^{(l)}$ with $l\ge 0$ and $\seq \xi1l \in
T_\Sigma(Z_k)$ we have
\[ h^{q_1 \ldots q_k}(\sigma(\seq \xi1l))_q = \sum_{p_1 \ldots p_l  \in Q^l}
\delta_l(p_1 \ldots p_l, \sigma, q) \cdot \prod_{i = 1}^k
h^{q_1 \ldots q_k}(\xi_i)_{p_i} \enspace. \]
\end{itemize}
We abbreviate $h^\varepsilon$ by $h$.
The weighted tree language \emph{recognized} by~$\cal A$, denoted also by~$\cal
A$, is the mapping $r_{\cal A} \colon T_\Sigma \to S$ defined for every
$\xi \in T_\Sigma$ by
\[ r_{\cal A}(\xi) = \sum_{q \in F} h(\xi)_q \enspace. \] A weighted tree language~$r: T_\Sigma \rightarrow S$ is \emph{recognizable} if
there exists a wta~${\cal A}$ such that $r_{\cal A} = r$. 

We note that a wta over a monadic ranked alphabet of input symbols is equivalent to an initial state normalized weighted automaton (as, e.g., used in \cite{bolgasmonzei10}). For a more detailed discussion about this special case we refer to  \cite[p.324]{fulvog09}.

We call a wta ${\cal A} = (Q, \Sigma, \delta,F)$ {\em final state normalized} if $|F| = 1$. 

\begin{lm}\cite[Thm.3.6]{fulvog09} \rm\label{lm:wta-normal}  For every wta there is an equivalent wta which is final state normalized.
\end{lm}

\subsection{Weighted Logics}

\paragraph{The Weighted MSO-Logic:}

Here we recall the weighted MSO-logic on trees which we will use in this paper. This weighted logic has its origin in \cite{drogas05,drogas07,drogas09} where it was defined for strings. It has been extended to trees in \cite{drovog06,fulvog09,drovog10a}. We present it in the form of \cite{bolgasmonzei10}.

 As usual in MSO-logic, we use first-order variables, like $x,x_1,x_2,\ldots, y,z$ and second-order variables, like $X,X_1,X_2, \ldots,Y,Z$.

We define the set of {\em weighted  MSO-logic formulas over $\Sigma$ and $S$}, denoted by  $\MSO(\Sigma,S)$ (or shortly: $\MSO$), to be the set of formulas generated by the following EBNF with nonterminal $\varphi$:
\begin{align*}
\varphi & ::=  a \mid
  \lab_\sigma(x)\mid \edge_i(x,y)\mid x\le y\mid x \in X \mid \\
& \hspace{6mm} \neg \varphi \mid \varphi \vee \varphi \mid \varphi \wedge
\varphi \mid \exists x.\varphi \mid \forall x.\varphi \mid \exists X.\varphi \mid \forall X.\varphi,
\end{align*}
\noindent where $a \in S$, $x,y$ are first-order variables, $\sigma \in \Sigma$, $1 \le i \le \maxrk(\Sigma)$, and
$X$ is a second-order variable. We will abbreviate a sequence $\exists x_1 \ldots \exists x_k$ of quantifications by $\exists x_{1,k}$. The set of free variables of a formula $\varphi$ is denoted by $\Free(\varphi)$. The formula $\varphi$ is called {\em sentence} if  $\Free(\varphi)=\emptyset$.
Often we indicate the free variables of a formula explicitly. For
instance, if a formula $\varphi$ has the free variables $x$, $y$, and
$z$, then we denote this fact by $\varphi(x,y,z)$. If $x_1,\ldots,x_k$
are the free variables of some formula $\psi$, then we write
$\psi(x_{1,k})$, and accordingly for other sequences of variables.

As usual in logics, we deal with free variables of a formula by means
of variable assignments. In the following we collect the most important notations and refer 
the reader to \cite{drogas05,drogas07,drogas09,drovog06,fulvog09,drovog10a,bolgasmonzei10}
for details. 

Let $\xi \in T_\Sigma$. For a finite set $\mathcal V$ of first-order and second-order
variables we denote a \emph{$\mathcal V$-assignment for $\xi$} by
$\rho$. For any  position $w \in \pos(\xi)$ and set $W \subseteq
\pos(\xi)$, we denote the $x$- and $X$-update of $\rho$ by
$\rho[x\mapsto w]$ and $\rho[X\mapsto W]$, respectively. 

In the usual way, we can encode a pair $(\xi,\rho)$, where $\rho$ is a
$\cal V$-assignment for $\xi$, as a tree $\zeta$ over the ranked alphabet
$\Sigma_\mathcal V$ with  $\Sigma_\mathcal V^{(k)}=\Sigma^{(k)}\times
\mathcal P(\mathcal V)$ for every $k\in\nat$.
	A tree $\zeta\in T_{\Sigma_\mathcal V}$ is called \emph{valid}
        if for every first-order variable $x\in \mathcal V$ there is a
        unique $w\in\pos(\zeta)$ such that $x$ occurs in the second
        component of $\zeta(w)$. We denote the set of all valid trees in $T_{\Sigma_\mathcal V}$ by $T_{\Sigma_\mathcal V}^\mathrm v$.

Let $\varphi \in \MSO$ and $\mathcal{V}$ be a finite set of
  variables containing $\Free(\varphi)$. The {\it semantics}
  of $\varphi$ is the weighted tree language $\seml \varphi
  \semr_\mathcal{V}: T_{\Sigma_\mathcal{V}} \rightarrow S$ defined as follows. If $\zeta \in  T_{\Sigma_\mathcal{V}}$ is not valid, then we put $\seml \varphi
  \semr_\mathcal{V}(\zeta) = 0$. Otherwise, we define $\seml \varphi
  \semr_\mathcal{V}(\zeta) \in S$ inductively as follows where $(\xi,\rho)$
  corresponds to $\zeta$.

{\small
\[
\begin{array}{ll}
\seml a \semr_\mathcal{V}(\zeta) = a &
\seml \varphi \vee \psi \semr_\mathcal{V}(\zeta)  = \seml \varphi \semr_\mathcal{V}(\zeta) + \seml \psi \semr_\mathcal{V}(\zeta)\\

\seml \mathrm{label}_\sigma(x) \semr_\mathcal{V}(\zeta)  = 
\left\{\begin{array}{ll}
1 & \hbox{if } \xi(\rho(x)) = \sigma,\\
0 & \hbox{otherwise}
\end{array}
\right.  &
\seml \varphi \wedge \psi \semr_\mathcal{V}(\zeta)  = \seml \varphi \semr_\mathcal{V}(\zeta) \cdot \seml \psi \semr_\mathcal{V}(\zeta)\\

\seml \mathrm{edge}_i(x,y) \semr_\mathcal{V}(\zeta)  = 
\left\{\begin{array}{ll}
1 & \hbox{if } \rho(y) =  \rho(x).i,\\
0 & \hbox{otherwise}
\end{array}
\right. &
\seml \exists x. \varphi  \semr_\mathcal{V}(\zeta)  = \hspace{-3mm} \sum\limits_{w \in \mathrm{pos}(\zeta)} \seml \varphi \semr_{\mathcal{V}\cup\{x\}}(\zeta[x  \mapsto w])\\

\seml x \le y \semr_\mathcal{V}(\zeta)  = 
\left\{\begin{array}{ll}
1 & \hbox{if }  \rho(x) \le \rho(y)\\
0 & \hbox{otherwise}
\end{array}
\right. &
\seml \exists X. \varphi  \semr_\mathcal{V}(\zeta)  = \hspace{-3mm}\sum\limits_{I\subseteq \mathrm{pos}(\zeta)} \seml \varphi \semr_{\mathcal{V}\cup\{X\}}(\zeta[X \mapsto I])\\

\seml x \in X \semr_\mathcal{V}(\zeta)  = 
\left\{\begin{array}{ll}
1 & \hbox{if } \rho(x) \in \rho(X),\\
0 & \hbox{otherwise}
\end{array}
\right. & 
	\seml \forall x. \varphi  \semr_\mathcal{V}(\zeta)  = \hspace{-3mm}\prod\limits_{w \in \mathrm{pos}(\zeta)} \seml \varphi \semr_{\mathcal{V}\cup\{x\}}(\zeta[x
	\mapsto w])\\

\seml \neg \varphi \semr_\mathcal{V}(\zeta) =
\left\{\begin{array}{ll}
1 & \hbox{ if } \seml \varphi \semr_\mathcal{V}(\zeta) = 0,\\
0 & \hbox{ otherwise }
\end{array}
\right. &
	\seml \forall X. \varphi  \semr_\mathcal{V}(\zeta)  = \hspace{-3mm}\prod\limits_{I \subseteq \mathrm{pos}(\zeta)} \seml \varphi \semr_{\mathcal{V}\cup\{X\}}(\zeta[X\mapsto I])
\end{array}
\]
}
The order of the factors in the product over $\pos(\zeta)$ is
arbitrary because $S$ is a commutative semiring. Let $\varphi$ be a formula with free variables $x_1,\ldots, x_n$, $\xi \in T_\Sigma$, and $\rho$ a $\Free(\varphi)$-assignment for $\xi$ such that $\rho(x_i)=u_i$ for $1\leq i \leq n$. Then we denote
the semiring element $\seml \varphi\semr(\zeta)$ by $\seml \varphi\semr(\xi,u_1,\ldots,u_n)$, where $\zeta=(\xi,\rho)$.

We abbreviate $\seml \varphi \semr_{\mathrm{Free}(\varphi)}$ by $\seml \varphi \semr$. We say that two formulas $\varphi$ and  $\psi$ with the same set of free variables are {\em equivalent}, and write $\varphi \equiv \psi$, if $\seml \varphi \semr = \seml \psi\semr$.
For any ${\cal L} \subseteq \MSO$, a weighted tree language $r:~T_\Sigma~\rightarrow~S$ is called {\em $\cal L$-definable} if there is a sentence $\varphi \in \cal L$ such that $r = \seml \varphi\semr$. 

A formula $\varphi$ is called {\em Boolean-valued} if $\{ \seml \varphi \semr_{\mathcal V}(\zeta) \mid \zeta  \in T_{\Sigma_{\cal V}}\}\subseteq \{0,1\}$ for every $\mathcal V$ containing $\Free(\varphi)$.  If $\seml \varphi\semr(\xi,u_1,\ldots,u_n)=1$ for some $\xi \in T_\Sigma$, and $u_1,\ldots, u_n\in \pos(\xi)$, then we abbreviate this fact by writing that
``$\varphi^\xi(u_1,\ldots,u_n)$ holds" or ``we have $\varphi^\xi(u_1,\ldots,u_n)$'' or just ``$\varphi^\xi(u_1,\ldots,u_n)$''.

For every $\varphi, \psi \in \MSO$, we define the macro $\varphi
\stackrel{+}{\rightarrow} \psi:= \neg \varphi \vee (\varphi \wedge
\psi)$. Then, for every $\mathcal V$ containing $\Free(\varphi)\cup \Free(\psi)$ and $\zeta  \in T_{\Sigma_{\cal V}}$, we have $\seml \varphi \stackrel{+}{\rightarrow} \psi \semr_{\cal V}(\zeta) =0$ if
$\zeta$ is not valid. If $\zeta$ is valid and $\varphi$ is
Boolean-valued, then  we have  that \label{p:+}
\[
\seml\varphi \stackrel{+}{\rightarrow} \psi\semr (\zeta) =
\left\{
\begin{array}{ll}
\seml\psi\semr(\zeta) & \hbox{ if } \seml \varphi\semr(\zeta) = 1\\
1 & \hbox{ if } \seml \varphi\semr(\zeta) = 0 \enspace.
\end{array}
\right.
\]
Clearly, if $\varphi$ and $\psi$ are Boolean-valued, then $\varphi \stackrel{+}{\rightarrow} \psi$ is Boolean-valued.

\paragraph*{The Boolean Fragment $\bMSO$:}

Next we define the Boolean fragment of $\MSO$  according to \cite{bolgasmonzei10}. The {\em Boolean fragment of $\MSO$}, denoted by $\bMSO$, is the set of all formulas generated by the EBNF
\begin{align*}
\varphi & ::=  0  \mid  1 \mid
  \lab_\sigma(x)\mid \edge_i(x,y)\mid x\le y\mid x \in X \mid  \neg \varphi \mid \varphi \wedge
\varphi \mid \forall x.\varphi \mid \forall X.\varphi\enspace.
\end{align*}
Clearly, every $\varphi \in \bMSO$ is Boolean-valued.

In $\bMSO$ we define the following macros: for every $\varphi, \psi \in \bMSO$ we let \label{p:B}
\[
 \varphi \underline{\vee} \psi := \neg(\neg \varphi \wedge \neg \psi), 
 \underline{\exists} x. \varphi := \neg \forall x. \neg \varphi, \text{ and }
\underline{\exists} X. \varphi := \neg \forall X. \neg \varphi. 
\]
Note that $\varphi\stackrel{+}{\rightarrow} \psi\equiv \neg \varphi \underline{\vee} (\varphi \wedge\psi)$.

\begin{ob}\label{ob:BMSO-rec} The semantics of any $\bMSO$-formula has the form $\mathds{1}_L$ where $L$ is a recognizable tree language.
\end{ob}
\begin{proof} Let $\varphi$ be a $\bMSO$-formula. Then 
$\varphi$ can be considered as a classical (unweighted) MSO-formula for
trees. If $\varphi$ does not contain the atomic formula $(x \le y)$,
then we obtain the result directly from Lemma 3.3(1) of
\cite{drovog06}. 

Now let $\varphi = (x \le y)$. Then it is easy to construct a wta
which checks the validity of the input tree; moreover, it checks
whether the position labeled by $x$ is a prefix of the $y$-labeled
position. It can perform the latter task  by switching into an alert state while
reading the $y$-labeled position, propagating the alert state, and
switching to the final state while reading the $x$-labeled position. 
\end{proof}

\paragraph*{${\cal L}$-step Formulas:}

Let ${\cal L} \subseteq \bMSO$ be closed under $\wedge$ and $\neg$. According to \cite{bolgasmonzei10},  the set of {\em ${\cal L}$-step formulas}, denoted by ${\cal L}_{\mathrm{step}}$, is the set of all $\MSO$-formulas generated by the EBNF
\begin{align*}
\varphi ::=  a \mid \alpha \mid \neg \varphi \mid \varphi \vee
\varphi \mid \varphi \wedge \varphi \hbox{ with $a \in S$ and $\alpha \in {\cal L}$}
\enspace.
\end{align*}

We will use the following technical result.

\begin{lm}\rm\label{Lemma3} For every  ${\cal L}$-step formula $\varphi$, there are $k \in \nat_+$, $a_1,\ldots,a_k\in S$, and $\varphi_1,\ldots,\varphi_k \in {\cal L}$ such that $\varphi \equiv \bigvee_{1 \le i \le k}(a_i \wedge \varphi_i)$. In particular, the semantics of $\varphi$ is a recognizable step function.
\end{lm}
\begin{proof} The first statement can be proved by an easy adaptation of  \cite[Lm. 3]{bolgasmonzei10}. The second statement follows from the first one, the definition of the semantics of MSO-formulas, and Observation \ref{ob:BMSO-rec}.
\end{proof}

\paragraph*{$\exists\forall({\cal L})$-Formulas:} Let  ${\cal L} \subseteq \MSO$. The fragment $\exists\forall({\cal L})$ consists of all $\MSO$-formulas of the form 
\[
\exists X. \forall x. \theta(X,x),
\]
where $\theta(X,x)\in {\cal L}$ and $\theta$ has the free variables $X$ and $x$. A weighted tree language $r: T_\Sigma \rightarrow S$ is {\em $\exists\forall({\cal L})$-definable} if there is a formula $\theta(X,x) \in {\cal L}$ such that for every $\xi \in T_\Sigma$
\[
r(\xi) = \seml \exists X. \forall x. \theta(X,x)\semr(\xi)\enspace.
\]

\paragraph*{The Fragment $\RMSO$:}

We define the fragment  $\RMSO$ of restricted MSO in the spirit of \cite{gas10,bolgasmonzei10}.  

Formally, the fragment $\RMSO$ is the set of all {\em weighted restricted MSO-formulas} generated by the EBNF:
\begin{align*}
\varphi & ::=  a \mid
  \lab_\sigma(x)\mid \edge_i(x,y)\mid x\le y\mid x \in X \mid \\
& \hspace{6mm} \neg \psi \mid \varphi \vee \varphi \mid \varphi \wedge
\varphi \mid \exists x.\varphi \mid \forall x.\psi \mid \exists X.\varphi \mid \forall X.\chi,
\end{align*}
where $\psi$ is a $\bMSO$-step formula and $\chi$ is a $\bMSO$-formula.

\paragraph*{The Fragments of First-Order Logic $\FO$  and $\bFO$:}
Another fragment of $\MSO$ is the set of {\em weighted first-order formulas over $\Sigma$ and $S$}, denoted by $\FO$, which is the set of all formulas generated by the EBNF
\begin{align*}
\varphi & ::=  a \mid
  \lab_\sigma(x)\mid \edge_i(x,y)\mid x\le y\mid x \in X \mid \neg \varphi \mid \varphi \vee \varphi  \mid \varphi \wedge
\varphi \mid \exists x.\varphi \mid \forall x.\varphi \enspace.
\end{align*}
Note that second-order variables may occur (as free variables).

The fragment $\bFO$ is defined to be the intersection $\bMSO \cap \FO$. That is, $\bFO$ is the set of all formulas generated by the EBNF
\begin{align*}
\varphi & ::=  0 \mid 1 \mid 
  \lab_\sigma(x)\mid \edge_i(x,y)\mid x\le y\mid x \in X \mid  \neg \varphi \mid  \varphi \wedge
\varphi \mid \forall x.\varphi \enspace.
\end{align*}

\paragraph*{The Fragment using Modulo Constraints $\bFOmod$:}

Let $n \in \nat_+$ and $m \in \nat$ such that $0 \le m < n$. We introduce the macro $|x| \equiv_n m$ with $x$ as the only free variable, and its intended meaning is as follows. For every tree $\xi \in T_\Sigma$ and $v \in \pos(\xi)$  we have \label{p:|x|equiv_n}
\begin{equation}
\seml |x| \equiv_n m \semr (\xi, v) =
\left\{
\begin{array}{ll}
1 & \hbox{ if } | v | \equiv m\, (\mbox{\hspace*{-3mm}}\mod n)\\
0 & \hbox{ otherwise,}
\end{array}
\right. \label{equ:modulo}
\end{equation}

We define $|x| \equiv_n m$ by the following $\bMSO$-formula: 
\begin{align*}
& (|x| \equiv_n m) := \\
& \forall X. \left(\left( (x\in X) \wedge \left(\forall y. ((y\in X) \wedge (|y| > n)) \stackrel{+}\rightarrow (y/n \in X)\right)\right)\stackrel{+}\rightarrow (m \in |X|)\right)
\end{align*}
where we use the following macros. For every $n\in \mathbb{N}$, let
\begin{itemize}
\item $(|y| > n) := \underline{\exists} x. \; (x \le_{n+1} y)$,
\item $x \le_n y := \underline{\bigvee}_{w \in \{1,\ldots,\maxrk(\Sigma)\}^n} x\le_w y$, \label{p:le_n}
\item $x\le_w y := \underline{\exists} y_{0,n}. (x=y_0) \wedge
\formpath_w(y_{0,n}) \wedge (y_n=y)$\\ for every $w = w_1 \ldots w_n \in \mathbb{N}^*$ with  $w_i \in \mathbb{N}$, \label{p:le_w}
\item $(x=y) := (x \le y) \wedge (y \le x)$,

\item $\formpath_w(y_{0,n}) := \bigwedge\limits_{1 \le i \le n}
\edge_{w_i}(y_{i-1},y_i)$,\\
($y_0,\ldots,y_n$  form a path via $w_1 \ldots w_n$).

\item $(y/n \in X) := \underline{\exists} x. \; (x\in X) \wedge (x
  \le_n  y)$,

\item $\edge(x,y) := \underline{\bigvee}_{1\leq i \leq \maxrk(\Sigma)} \edge_i(x,y)$,

\item $\rt(x) :=  \forall y. \neg\,\edge(y,x)$,

\item $(m \in |X|) := \underline{\exists} x,y. \; \rt(x) \wedge (x
  \le_m y) \wedge (y \in X)$\enspace.
\end{itemize}
It is easy to see that (\ref{equ:modulo}) holds.

Let us denote by $\bFOmod$ the fragment of $\bMSO$ which we obtain by adding the formula $|x| \equiv_n m$ for every $n \in \nat_+$ and $m \in \nat$  to the list of the alternatives defining the fragment $\bFO$.

\section{Branching Transitive Closure}\label{bTC-section}

In this section we introduce our branching transitive closure operator
$\BTC$, and we define its application $\BTC(\Phi)$ where $\Phi$ is a finite family of formulas of the form $\varphi_k(x,y_1,\ldots,y_k)$ with one free input variable $x$ and $k$ free output variables $y_1,\ldots,y_k$. We require that the formulas of $\Phi$ satisfy a certain progress which is determined by a natural number $n \in \nat_+$. 

The progress is defined in terms of base positions. For every $\xi \in T_\Sigma$ and position $v \in\pos(\xi)$ there is a uniquely determined prefix $u$ of $v$ such that $|u| = i\cdot n$ for some $i \in \mathbb{N}$ and $0 \le |v| - |u| < n$. We call this $u$ the {\em base position of $v$} and denote it by $\langle v\rangle$.

Then, intuitively, positions $u$ and $u_{1,k}$ constitute a progress if (a) for every $1 \le i \le k$, the base positions $\langle u\rangle$ and $\langle u_i\rangle$ have distance $n$ and (b) $\langle u_1\rangle, \ldots, \langle u_k\rangle$ are siblings ordered from left to right (cf. Figure \ref{fig:progress} for $k=2$). Formally, we define the {\em $n$-progress formula} $\psi_{k}(x,y_{1,k})$ for every $k\in \nat$ as follows:
\[
\begin{array}{rcl}
\psi_{k}(x,y_{1,k}) &:=&  \underline{\exists} z,z_{1,k}.
(z=\langle x\rangle_n) \wedge \bigwedge_{1\le i\le k} (z_i=\langle y_i\rangle_n)\\
&& \wedge
\bigwedge_{1 \le i \le k} (z \le_n z_i) \wedge
\bigwedge_{1 \le i \le k-1} \sibl_n(z_i,z_{i+1}),
\end{array}
\]
where \label{p:BTC}
\begin{itemize}
\item $(y=\langle x\rangle_n) := \bigwedge_{0\leq q < n} \big((|x|
  \equiv_n q) \stackrel{+}\rightarrow (y \le_q x)\big)$

\item $\sibl_n(x, y) := 
\underline{\exists} x',y'. \big( \sibl(x', y') \wedge \underline{\bigvee}_{1\leq l\leq n-1}((x'\le_{l} x) \wedge
(y' \le_{l} y))\big)$
\item $\sibl(x, y) := \underline{\exists} z. \underline{\bigvee}_{1\leq i < j \leq \maxrk(\Sigma)} \edge_i(z,x)\wedge \edge_j(z,y)$\\
 ($x$ is a younger sibling of $y$).
\end{itemize}
\begin{figure}
\begin{center}
\includegraphics[scale=0.97]{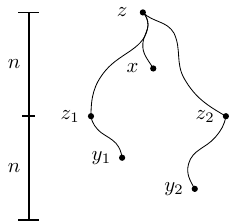}
\end{center}
\caption{\label{fig:progress} Illustration of a progress for $k=2$.}
\end{figure} 
We note that $\psi_{k}(x,y_{1,k})$ is in BFO+mod and it is irreflexive in the following sense.
Since $n \ge 1$ we have that for every $\xi \in T_\Sigma$ and $v,v_{1,k} \in \pos(\xi)$: if $\psi_{k}^\xi(v,v_{1,k})$ holds, then $|v| < |v_i|$ for every $1 \le i \le k$. Note that, in general, $v$ is not a prefix of $v_i$. Finally, we note that  $\psi_{k}^\xi(v,v_{1,k})$ implies that $\psi_{1}^\xi(v,v_i)$ for every $1\le i \le k$.

Let ${\cal L} \subseteq \MSO$, $m \in \mathbb{N}$ and $n \in \mathbb{N}_+$. An {\em $m$-family of  formulas in $\cal L$} is a family 
$$\Phi = (\varphi_k(x,y_{1,k}) \mid 0 \le k \le m),$$ 
where  $\varphi_k(x,y_{1,k})  \in {\cal L}$ for every $0 \le k \le m$. Moreover $\Phi$  is {\em $n$-progressing} if for every $0 \le k \le m$, tree $\xi \in T_\Sigma$, and $v,v_1,...,v_k \in \pos(\xi)$: 
\begin{equation}
\seml\varphi_{k}(x,y_{1,k})\semr(\xi,v,v_{1,k}) \neq 0 \text{ implies } \psi_{k}^\xi(v,v_{1,k})\enspace. \label{equ:progress}
\end{equation}
In particular, $\Psi=(\psi_{k}(x,y_{1,k}) \mid 0 \le k \le m)$ is an $m$-family of $n$-progressing  formulas in BFO+mod. \label{Psi-def}

For the definition of our BTC operator, we need a source of infinitely many fresh first-order variables. Therefore we specify a first-order variable $x_w$ for every $w\in \mathbb{N}_+^*$ such that $w\neq w'$ implies $x_w\neq x_{w'}$. For every $w\in \mathbb{N}_+^*$ and $k\in \mathbb{N}_+$, we abbreviate the sequence $x_{w1},\ldots,x_{wk}$ by $x_{w1,k}$. For $k=0$ we define $x_{w1,k}$ to be the empty sequence.

Now let $\Phi$ be an $m$-family of $n$-progressing formulas in $\cal L$ and $w\in \mathbb{N}_+^*$.  
We  define the family $(\BTC^{l}_w(\Phi) \mid l \ge 1)$ of $\MSO$-formulas by induction as follows:
\begin{enumerate}
\item[(i)] $\BTC^{1}_w(\Phi) = \varphi_0(x_w)$  
\item[(ii)] $\BTC^{l+1}_w(\Phi) = \bigvee\limits_{1 \le k \le m} \exists
  x_{w1,k}. \; 
  \varphi_k(x_w,x_{w1,k}) \wedge\bigvee\limits_{\substack{l_1,...,l_k \in \nat_+\\l_1+...+l_k = l}}
\bigwedge\limits_{1 \le i \le k} 
 \BTC^{l_i}_{wi}(\Phi)$\enspace.
\end{enumerate}
Notice that $\varphi_k(x_w,x_{w1,k})$ abbreviates the formula $\varphi_k(x/x_w,y_1/x_{w1},\ldots,y_k/x_{wk})$ obtained by variable substitution. Moreover, $x_w$ is the only free variable of $\BTC^{l}_w(\Phi)$. 

For instance, consider the family $\Phi =
(\varphi_0(x), \varphi_1(x,y_1), \varphi_2(x,y_1,y_2))$. Then we have
{\small
\[
\begin{array}{rclcl}
\BTC_\varepsilon^4(\Phi) & = &
\Big(\exists x_1. \; \varphi_1(\xe,x_1) & \hspace{-12mm} \wedge &
    \hspace{-8mm} [( \exists x_{11}. \varphi_1(x_1,x_{11}) \wedge \exists x_{111}. \varphi_1(x_{11},x_{111})  \wedge  \varphi_0(x_{111}))\\
& &  &  &\hspace{-11mm} \vee \;\;  (\exists x_{11}, x_{12}. \varphi_2(x_{1},x_{11},x_{12})  \wedge \varphi_0(x_{11})  \wedge  \varphi_0(x_{12})) ]\Big)\\[3mm]
&  \vee 
& \Big( \exists x_1,x_2. \varphi_2(\xe,x_1,x_2) & \wedge & [ (\exists x_{11}. \varphi_1(x_1,x_{11}) \wedge \varphi_0(x_{11})
       \wedge \varphi_0(x_2))\\
& & & &  \vee \;\; (\varphi_0(x_1) \wedge \exists
x_{21}. \varphi_1(x_1,x_{21}) \wedge \varphi_0(x_{21}))] \Big)
\end{array}
\]
}

The {\em branching transitive closure of $\Phi$ } is just
the expression $\BTC(\Phi)$. The {\em semantics of  $\BTC(\Phi)$},
denoted by $\seml \BTC(\Phi) \semr$, is the mapping defined for every $\xi\in
T_\Sigma$ and $v\in \pos(\xi)$ by
$$\seml \BTC(\Phi)\semr(\xi,v)=\seml \bigvee_{1\leq l\leq \size(\xi)} \BTC^{l}_\varepsilon(\Phi)\semr (\xi,v).$$
We note that it suffices to let $l$ range over the finite set
$\{1,\ldots,\size(\xi)\}$, because the progress formula
$\psi_{k}(x,y_{1,k})$ is irreflexive and implication
\eqref{equ:progress} holds. Hence we have $\seml \BTC^{l}_\varepsilon(\Phi)\semr (\xi, v) =0$ for every $\xi \in T_\Sigma$ and $v\in \pos(\xi)$, provided that $l> \size(\xi)$.

Moreover, we define $\BTC({\cal L})$ to be the class of all expressions of
the form $\BTC(\Phi)$, where $\Phi$ is an $m$-family of $n$-progressing formulas in $\cal L$
for some $m \in \mathbb{N}$ and $n \in \mathbb{N}_+$.
Finally, a weighted tree language  $r: T_\Sigma \rightarrow S$  is 
 {\em $\BTC({\cal L})$-definable}  if there is an expression
 $\BTC(\Phi)$ in $\BTC({\cal L})$ such that for every $\xi \in T_\Sigma$:
\[
r(\xi) = \seml \BTC(\Phi) \semr (\xi,\varepsilon).
\]


Before showing an example of a $\BTC({\cal L})$-definable weighted tree language, we take a slightly different point of view to the underlying formulas; this will be helpful in Section \ref{sect:BTC-EA}.

The MSO-formula $\BTC_w^l(\Phi)$ contains a number of scattered occurrences of disjunction, where each occurrence has one of the following two forms: (1) disjunction of the form  ``$\bigvee_{1 \le k \le m}$'' for the choice of a rank $k$ or (2) disjunction of the form ``$\bigvee_{\substack{l_1,...,l_k \in \nat_+\\l_1+...+l_k = l}}$''
for the choice of a partitioning of $l$ into summands $l_1,\ldots,l_k$. Instead of having these disjunctions scattered over the whole formula, we could pull them out and make all the choices in advance. This leads to the notion of unfolding. 

Formally, let $\Phi=(\varphi_{k}\mid 0\le k\le m)$ be again an $m$-family of $n$-progressive formulas in $\cal L$ and $w\in \mathbb{N}_+^*$. For every $l\in \nat_+$, we define
the {\em set of unfoldings of $\BTC_w^l(\Phi)$}, denoted by  $\uBTC_w^l(\Phi)$, by induction on $l$:
\begin{enumerate}
\item[(i)]  $\uBTC_w^1(\Phi)=\{\varphi_0(x_w)\}$,
\item[(ii)] $\uBTC_w^{l+1}(\Phi)=  \bigcup\limits_{1\le k\le m}\big\{\exists x_{w1,k}.\varphi_k(x_w,x_{w1,k})\wedge \chi_1 \wedge \ldots \wedge \chi_k \mid$\\
$\hspace*{45mm} \forall i: \chi_i \in \uBTC_{wi}^{l_i}(\Phi) \text{ such that } l_1+\ldots + l_k=l \big\}.$
\end{enumerate}
Again, let $\Phi =
(\varphi_0(x), \varphi_1(x,y_1), \varphi_2(x,y_1,y_2))$. Then we have
{\small
\[
\begin{array}{rcl}
\uBTC_\varepsilon^4(\Phi) & = \Big\{ & 
 \exists x_1. \; \varphi_1(\xe,x_1)  \wedge \big(
    \exists x_{11}. \varphi_1(x_1,x_{11}) \wedge (\exists x_{111}. \varphi_1(x_{11},x_{111})  \wedge  \varphi_0(x_{111}))\big), \\[2mm]
& & \exists x_1. \; \varphi_1(\xe,x_1)  \wedge \big(
\exists x_{11}, x_{12}. \varphi_2(x_{1},x_{11},x_{12})  \wedge \varphi_0(x_{11})  \wedge  \varphi_0(x_{12})\big),  \\[2mm]
&  & \exists x_1,x_2. \varphi_2(\xe,x_1,x_2) \wedge \big(\exists x_{11}. \varphi_1(x_1,x_{11}) \wedge \varphi_0(x_{11})\big)
       \wedge \varphi_0(x_2), \\[2mm]
&  &  \exists x_1,x_2. \varphi_2(\xe,x_1,x_2)  \wedge 
\varphi_0(x_1) \wedge \big(\exists
x_{21}. \varphi_1(x_1,x_{21}) \wedge \varphi_0(x_{21})\big)
 \Big\}
\end{array}
\]
}
Note that each formula $\uBTC^{l}_w(\Phi)$ is rectified and that $x_w$ is its only free variable.

Due to the distributivity of multiplication over addition in the semiring $S$, we can easily prove the following connection between $\BTC_w^l(\Phi)$ and its set of unfoldings by induction on $l$.

\begin{ob} \rm \label{lm:alternative-semantics} For every $l\in \nat_+$, $w\in \mathbb{N}_+^*$ $\xi\in T_\S$, and $v \in\pos(\xi)$,  we have
\[\seml \BTC_w^l(\Phi)\semr (\xi,v) = \sum\limits_{\chi \in \uBTC_w^l(\Phi)}\seml \chi\semr (\xi,v).\]
\end{ob}

For a formula $\varphi \in \MSO$, we denote by
$\overline{\varphi}$ the formula obtained by deleting all quantifications $\exists x$ from $\varphi$. 

For every $\chi \in \uBTC_w^l(\Phi)$, the formula $\overline{\chi}$ is a conjunction of $l$ formulas taken from $\Phi$; all variables occurring in $\overline{\chi}$ are free; in fact, $\overline{\chi}$ has $l$ free variables.  Figure \ref{fig:unfolding-fig} illustrates $\overline{\chi}$ and a particular assignment of positions (shown as solid bullets) to its free variables. The base positions of these positions are indicated by circles and the tree structure of these base positions is indicated by lines.

\begin{figure}[h]
\begin{center}
\includegraphics[scale=0.87]{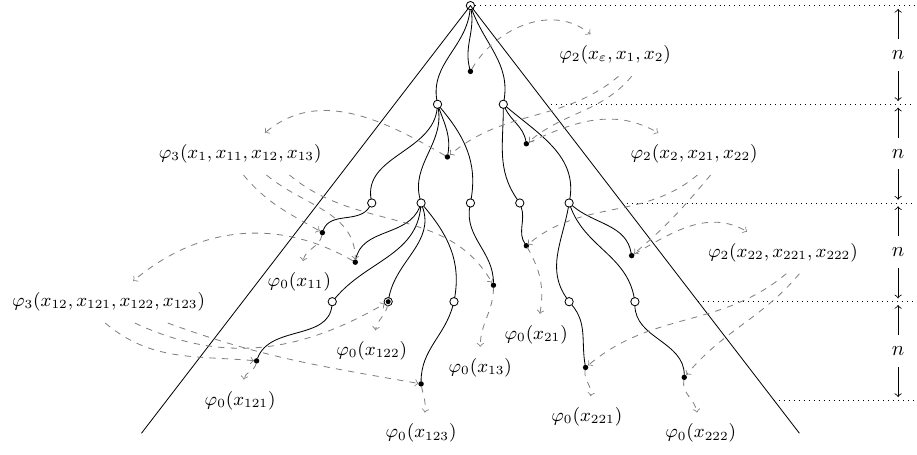} 
\end{center}
\caption{\label{fig:unfolding-fig} An example unfolding in $\overline{\chi} \in \uBTC_\varepsilon^{13}(\Phi)$.} 
\end{figure}

\begin{ex}\rm  1) Let $\Sigma = \{\delta^{(3)}, \alpha^{(0)},
  \beta^{(0)}\}$ and consider the semiring $(\nat,+,\cdot,0,1)$ of
  natural numbers. Let $L$ be the set of all trees  generated by the regular tree grammar with the two rules
\[S \rightarrow \beta \; \text{ and } \; S \rightarrow \delta(S,\alpha,S)\enspace.\]
We now want to define a family $\Phi$ of formulas such that $\seml \BTC(\Phi)\semr$ is the characteristic mapping of $L$, i.e.,
\[
\seml \BTC(\Phi)\semr(\xi,\varepsilon) = 
\left\{
\begin{array}{ll}
1 & \text{ if }\xi \in L \\
0 & \text{ otherwise.}
\end{array}
\right.
\]
For this we define the 2-family $\Phi = (\varphi_k(x,y_{1,k}) \mid 0 \le k \le 2)$  by
\begin{align*}
\varphi_0(x) & =  \lab_\beta(x),\\
\varphi_1(x,y_1) & =  \mathrm{false}, \;\;\text{ and }\\
\varphi_2(x,y_1,y_2) & =  \lab_\delta(x) \wedge \edge_1(x,y_1) \wedge 
\Big(\exists z. \edge_2(x,z) \wedge \lab_\alpha(z)\Big) \wedge \edge_3(x,y_2)\enspace.
\end{align*}
Note that for every $\xi \in T_\Sigma$ and $u,v \in \pos(\xi)$ we have $\seml y=\langle x \rangle_1 \semr(\xi,u,v) = 1$ implies $u=v$. Thus for the 1-progress formulas $\psi_{0}$, $\psi_{1}$, and $\psi_{2}$ we obtain the following equivalences:
\begin{align*}
\psi_{0}(x)& \equiv \mathrm{true},\\ 
\psi_{1}(x,y_1)& \equiv (x \le_1 y_1), \;\; \text{ and }\\
\psi_{2}(x,y_1,y_2)& \equiv (x \le_1y_1)\wedge 
(x \le_1y_2)\wedge\sibl(y_1,y_2)\enspace.
\end{align*}
Since for every $\xi \in T_\Sigma$ and $v,v_1,v_2 \in \pos(\xi)$: 
\begin{align*}
\seml\varphi_{0}(x)\semr(\xi,v) \neq 0 & \text{ implies } \psi_{0}^\xi(v)\\
\seml\varphi_{1}(x,y_1)\semr(\xi,v,v_1) \neq 0 & \text{ implies } \psi_{1}^\xi(v,v_1)\\
\seml\varphi_{2}(x,y_1,y_2)\semr(\xi,v,v_1,v_2) \neq 0 & \text{ implies } \psi_{2}^\xi(v,v_1,v_2)\enspace,
\end{align*}
we have that $\Phi$ is a  2-family of 1-progressing formulas.

By induction on $l$ we can show the following statement: for every $l \ge 1$, $w\in \mathbb{N}_+^*$, $\xi \in T_\Sigma$, and $v  \in \pos(\xi)$
\begin{equation}
\seml \BTC_w^l(\Phi)\semr(\xi,v) =
\left\{
\begin{array}{ll}
1 & \text{ if }\xi|_v \in L \text{ and } l = |\pos_{\{\delta,\beta\}}(\xi|_v)| \\
0 & \text{ otherwise.}
\end{array}
\right. \label{eq:ex}
\end{equation}
Then, due to the definitions and using \eqref{eq:ex} with $v=\varepsilon$ we have:
\[
\seml \BTC(\Phi)\semr(\xi,\varepsilon) = 
\sum_{1 \le l \le \mathrm{size}(\xi)} \seml \BTC_\varepsilon^l(\Phi)\semr(\xi,\varepsilon) = 
\seml \BTC_\varepsilon^{|\pos_{\{\delta,\beta\}}(\xi)|}(\Phi)\semr(\xi,\varepsilon)\enspace. 
\]
Finally, again using \eqref{eq:ex} we obtain that $\seml \BTC(\Phi)\semr$ is the characteristic mapping of $L$ in the above sense. 

We note that the family $\Phi$ traverses the given tree $\xi$ vertically maximal in the sense that the iteration has to stop at $\beta$-labeled leaves.

2) In our second example we consider the ranked alphabet $\Sigma = \{\sigma^{(2)}, \alpha^{(0)}\}$ and the semiring of natural numbers.
We define the 2-family  $\Phi' = (\varphi_k(x,y_{1,k}) \mid 1 \le k \le 2)$ of 1-progressing  formulas by
\begin{align*}
\varphi_0(x) & =  \mathrm{true},\\
\varphi_1(x,y_1) & =  \mathrm{false}, \;\; \text{ and }\\
\varphi_2(x,y_1,y_2) & =  \lab_\sigma(x) \wedge \edge_1(x,y_1) \wedge \edge_2(x,y_2)\enspace.
\end{align*}
Moreover, for each $\xi \in T_\Sigma$, we define the set $\mathrm{prefix}_l(\xi)$ of $l$-prefixes  of $\xi$ to be the set of all
``top parts'' of $\xi$ which contain $l$ occurrences of the symbol $\sigma$. Formally, let $\Delta =\{ \sigma^{(2)}, *^{(0)} \}$ and
for every $l \geq 1$, we define
$$\mathrm{prefix}_l(\xi) = \{ \zeta \in T_\Delta \mid \pos(\zeta) \subseteq \pos(\xi) \text{ and }| \pos_\sigma(\zeta) | = l \}.$$
Let, for instance, $\xi= \sigma(\sigma(\alpha,\alpha),\sigma(\alpha,\alpha))$. Then $\mathrm{prefix}_1(\xi) = \{\sigma(*,*)\}$,
 $\mathrm{prefix}_2(\xi) = \{\sigma(\sigma(*,*),*), \sigma(*,\sigma(*,*)) \}$,  $\mathrm{prefix}_3(\xi) = \{\sigma(\sigma(*,*),\sigma(*,*)) \}$, and $\mathrm{prefix}_l(\xi) =\emptyset$ for every $l\geq 4$.

Then, for every $\xi \in T_\Sigma$, $l \ge 1$, and $w\in \mathbb{N}_+^*$, we have
\[
\seml \BTC_w^{2l+1}(\Phi')\semr(\xi,\varepsilon) = |\mathrm{prefix}_l(\xi)|.
\]
because every element of $\mathrm{prefix}_l(\xi)$ can be identified with
a $2l+1$-fold iteration  of $\BTC_w$ on $\xi$. (The factor $2l+1$ comes from the fact that each  element
of $\mathrm{prefix}_l(\xi)$ has $2l+1$ nodes.)
An iteration of the $\BTC_w$-operator is not any more vertically maximal, because it can also stop at inner positions. Intuitively speaking, an iteration only spans a prefix of the tree (starting from its root). Thus
\[
\seml \BTC(\Phi')\semr(\xi,\varepsilon) = \sum_{1 \le l \le \size(\xi)}|\mathrm{prefix}_l(\xi)|\enspace.
\]
\end{ex}

\section{The Main Result}

\begin{theo}\label{main} Let $S$ be an arbitrary commutative semiring and $r: T_\Sigma \rightarrow S$ a weighted tree language. Then the following are equivalent:
\begin{enumerate}
\item[(a)] $r$ is recognizable,
\item[(b)] $r$ is $\BTC((\bFOmod)_\mathrm{step})$-definable,
\item[(c)] $r$ is $\BTC(\bMSO_\mathrm{step})$-definable,
\item[(d)] $r$ is $\exists\forall((\bFOmod)_\mathrm{step})$-definable,
\item[(e)] $r$ is $\exists\forall(\bMSO_\mathrm{step})$-definable,
\item[(f)] $r$ is $\RMSO$-definable.
\end{enumerate}
\end{theo}
\begin{proof} Theorem \ref{BbTC-theo}  proves that (a) implies (b).
By Theorem~\ref{EA-theo}, (b) implies (d) and (c) implies (e). 
Since $\bFOmod \subseteq \bMSO$, we have that (b) implies (c), and (d) implies (e).
Since $\exists\forall(\bMSO_\mathrm{step}) \subseteq \RMSO$, also (e) implies (f).   
By Theorem \ref{th:rec=definable} (f) implies (a). 
\end{proof}

As a corollary of our main result, we obtain a characterization of recognizable tree languages
in terms of our branching transitive closure operator. Let us denote
by $\MSO_t(\Sigma)$ (or shortly by $\MSO_t$) the set of (unweighted) monadic second order formulas for trees over $\Sigma$ (cf. \cite{don70,thawri68}) and by $\FO_t$ its first order segment. 

\begin{cor}\label{main-cor} Let $L\subseteq T_\Sigma$ be an arbitrary  tree language. Then the following are equivalent:
\begin{enumerate}
\item[(a)] $L$ is recognizable,
\item[(b)] $L$ is $\BTC(\FOtmod)$-definable,
\item[(c)]  $L$ is $\BTC(\MSO_t)$-definable,
\item[(d)] $L$ is $\exists\forall(\FOtmod)$-definable,
\item[(e)] $L$ is $\exists\forall(\MSO_t)$-definable,
\item[(f)] $L$ is $\MSO_t$-definable.
\end{enumerate}
\end{cor}
\begin{proof}  (Sketch.) Since Theorem \ref{main} holds for the Boolean semiring $\mathbb{B}$
(with operations disjunction and conjunction), it suffices to prove the following statement $(\dagger)$:
for every  $L\subseteq T_\Sigma$ and $\mathit{x} \in \{  \mathit{a},\ldots,\mathit{f}\}$, statement $\mathit{(x)}$ holds for $L$ if and only if statement $\mathit{(x)}$ of  Theorem \ref{main} holds for $S=\mathbb{B}$ and $r=\mathds{1}_L$.

For the proof of $x=a$, see   \cite[Subsect. 3.2]{fulvog09}.

To prove case $x=f$, first we observe that the logics $\RMSO(\Sigma,\mathbb{B})$ and $\MSO(\Sigma,\mathbb{B})$ are equivalent. Moreover, each $\MSO_t(\Sigma)$-formula can be considered as an $\MSO(\Sigma,\mathbb{B})$-formula with the same semantics. Vice versa, every $\MSO(\Sigma,\mathbb{B})$-formula can be transformed into an equivalent $\MSO_t(\Sigma)$-formula  by writing, e.g., $\exists x. (\lab_\sigma(x) \wedge \neg \lab_\sigma(x))$ for 0 and  $\forall x. (\lab_\sigma(x) \vee \neg \lab_\sigma(x))$ for 1 for some $\sigma \in \Sigma$. Hence $(\dagger)$ holds in this case.

To prove case $x=b$, we observe that $(\bFOmod)_\mathrm{step}(\Sigma,\mathbb{B})$ and $\FOmod(\Sigma,\mathbb{B})$ are equivalent. Moreover, $\FOtmod(\Sigma)$-formulas and $\FOmod(\Sigma,\mathbb{B})$-formulas correspond to each other in the natural way described above for $\MSO_t(\Sigma)$ and $\MSO(\Sigma,\mathbb{B})$.  The proof of the cases $x \in \{c, d,e\}$ are similar.
\end{proof}

\section{From wta To Branching Transitive Closure}
\label{sec:construction}

In this section we will simulate the behaviour of a wta by the 
branching transitive closure of a particular family of formulas.  
 Our goal is the following theorem.

\begin{theo}\label{BbTC-theo} \rm For every wta ${\cal A}$ with $n$
  states and input alphabet $\Sigma$ there is an $m \in \mathbb{N}_+$ and
  an $m$-family $\Phi_{\cal A}$ of $n$-progressing  formulas in $(\bFOmod)_\mathrm{step}$
  such that  $r_{\cal A}(\xi)~=~\seml \BTC(\Phi_{\cal A})
  \semr(\xi,\varepsilon)$ for every $\xi \in T_\Sigma$. 
\end{theo}

In this section we assume that ${\cal A} = (Q,\Sigma,\delta,F)$ is a wta with $Q =
\{0,\ldots,n-1\}$ and $n\in \nat_+$. By Lemma \ref{lm:wta-normal} we
can assume that $F = \{0\}$.

The main idea behind the following construction and the inductive proof of Theorem \ref{BbTC-theo} (cf.
Statement 1 in the proof of this theorem) is due to \cite{tho82}. First we decompose an input tree $\xi$ into  slices (cf. Section \ref{sec:slices}). The number $n$ and the ranked alphabet $\Sigma$ determine the maximal width of slices which we denote by $\max(\Sigma,n)$.
The behaviour of ${\cal A}$ on $\xi$ induces a behaviour  on the slices of $\xi$ (cf. Lemma~\ref{lm:decomp}).
Then we construct  $\Phi_{\cal A} = (\varphi_k(x,y_{1,k}) \mid 0 \le k \le \max(\Sigma,n))$ such that the behaviour of ${\cal A}$ on slices
is simulated  by $\BTC(\Phi_{\cal    A})$. More precisely,  let us denote the topmost slice of the decomposition of  $\xi$ at  some position $u$ by  $\head_n(\xi,u)$ and the positions   of $\xi$ at which the slices below  $\head_n(\xi,u)$ start, by $u_1\ldots u_k$. Then we construct $\Phi_{\cal A}$ such that  the decomposition 
\begin{equation}
h(\xi|_u)_q = \sum_{q_1,\ldots,q_k \in Q} h^{q_1\ldots
  q_k}(\head_n(\xi,u))_q \cdot \prod_{1 \le i \le k}
h(\xi|_{u_i})_{q_i} \label{equ:synch-h}
\end{equation}
of the behaviour of ${\cal A}$ is synchronized with one level of the iteration 
\begin{equation}
 \BTC^{l+1}_w(\Phi_{\cal A}) = 
 \bigvee\limits_{1 \le k \le \max(\Sigma,n)} \exists x_{w1,k}. \; 
  \varphi_k(x_w,x_{w1,k}) \wedge 
\bigvee\limits_{\substack{l_1,...,l_k \in \nat_+\\l_1+...+l_k = l}}
\bigwedge\limits_{1 \le i \le k}
 \BTC^{l_i}_{wi}(\Phi_{\cal A}) \label{equ:synch-phi} 
\end{equation}
for some $w\in \nat_+^*$ and such that each $\varphi_k$ is an $n$-progressing formula.

In Fig. \ref{fig:synchronization} we visualize this 
synchronization for a wta $\cal A$ with state set $\{0,1,2\}$. In part (a) we show the subexpression $h^{20}(\head_3(\xi,u))_1 \cdot 
h(\xi|_{u_1})_2\cdot h(\xi|_{u_2})_0$ of the right-hand side of
  \eqref{equ:synch-h} with $n=3$, $k=2$, $q=1$, $q_1=2$, and
  $q_2=0$. In part (b) we visualize the synchronization.

\begin{figure}
  \centering
  \begin{subfigure}[b]{0.49\textwidth}
	\centering
	\begin{tikzpicture}[remember picture]
	\draw (0,0) node[offpath] (ur) {} 
					to (-\xfull,\ybot) to[out = 230, in = 180] (-\xfull / 2,\ybot + \ymiddiff) node[offpath] (u1) {}
					to[out = 0, in = 180] (0, \ybot)
					to[out = 0, in = 180] (\xfull / 2, \ybot + \ymiddiff) node[offpath] (u2) {}
					to[out = 0, in = 310] (\xfull, \ybot)
					to (0,0);
	\node at (-\xfull,\ybot/4) {$h^{\,\tikz[mathcoord] \node (2){\scriptsize$2$}; ~ \tikz[mathcoord] \node (0){\scriptsize$0$};}~\bigg(_{\vphantom{\tikz[mathcoord] \node {$1$};}}$};
	\node at (\xfull - .4em,\ybot/4){$\bigg)_{\tikz[mathcoord] \node (1) {\scriptsize$1$};}$};
	\node[labeloffset] at (u1) {$u_1$};
	\node[labeloffset] at (u2) {$u_2$};
	\node[labeloffset] at (ur) {$u$};
	\node[anchor = east] at (\xfull/2, \ybot/2) {\scriptsize$\head_{3}(\xi,u)$};
	\path[dashed] (2) edge (u1)
							 (0) edge (u2)
							 (1) edge (ur);

	\begin{scope}[yshift = 1.3 * \ybot, xshift = -\xfull/2 - 0.8em]
		\draw (0,0) node (xi1) {}
			to (-\xfull/2, \ybot)
			to (\xfull/2, \ybot * 4 / 5)
			to (0,0);
		\node at (-\xfull* 4 / 9,\ybot/4) {$h\bigg($};
		\node at (\xfull* 4 / 9,\ybot/4) {$\bigg)_2$};
		\node at (\xfull / 3 - .3em, \ybot * 3/4 + .3em) {$\xi|_{u_1}$};
	\end{scope}

	\begin{scope}[yshift = 1.3 * \ybot, xshift = \xfull/2 + 0.8em]
		\draw (0,0) node (xi2) {}
			to (-\xfull/2, \ybot * 4.5 / 6)
			to (\xfull/2, \ybot * 4.5 / 6)
			to (0,0);
		\node at (-\xfull * 4 / 9,\ybot/4) {$h\bigg($};
		\node at (\xfull* 4 / 9,\ybot/4) {$\bigg)_0$};
		\node[xshift = -.2em, yshift = .03em] at (\xfull / 3, \ybot * 2/3) {$\xi|_{u_2}$};
	\end{scope}

	\path[->, thick, shorten > = .4em, shorten < = .2em] 
		(xi1) edge (u1)
		(xi2) edge (u2); 
	\end{tikzpicture}
	\caption{}
  \end{subfigure}
  \begin{subfigure}[b]{0.49\textwidth}
	\centering
	\begin{tikzpicture}[remember picture]
	\draw (0,0) node[onpath] (ur) {} 
					to (-\xfull,\ybot) to[out = 230, in = 180] (-\xfull / 2,\ybot + \ymiddiff) node[onpath] (u1) {}
					to[out = 0, in = 180] (0, \ybot)
					to[out = 0, in = 180] (\xfull / 2, \ybot + \ymiddiff) node[onpath] (u2) {}
					to[out = 0, in = 310] (\xfull, \ybot)
					to (0,0);
	
	\node at (\xfull + 0.4em ,\ybot * 4/9){$\varphi(\tikz[mathcoord]\node (xw) {$x_w$};, \tikz[mathcoord] \node (xw1) {$x_{w1}$};, \tikz[mathcoord] \node (xw2) {$x_{w2}$};)$};
	\node[onpath] at ($(u1) !.5! (ur)$) (u3) {};
	\node[onpath, xshift = -\xfull / 4] at (u1) (u4) {};
	\node[labeloffset] at (u1) {$u_1$};
	\node[labeloffset] at (u2) {$u_2$};
	\node[labeloffset] at (ur) {$u$};
	\node[anchor = east, shift={(.2em, -0.3em)}] at (\xfull/2, \ybot/2) {\scriptsize$\head_{3}(\xi,u)$};

	\begin{scope}[shift = {(u1)}]
		\draw (u1) node (xi1) {}
			to (-\xfull/2, \ybot)
			to (\xfull/2, \ybot * 4 / 5)
			to (u1);
		\node at (\xfull / 3 - .3em, \ybot * 3/4 + .3em) {$\xi|_{u_1}$};
		\node[onpath] at (0, \ybot * 1/3) (u5) {};
		\node[onpath] at (-\xfull / 8, \ybot * 2/3) (u6) {};
		\node at (0, 1.15 * \ybot) (btc1) {$\BTC^{\ell_1}_{w1}(\Phi_{\mathcal{A}})$};
	\end{scope}

	\begin{scope}[shift ={(u2)}]
		\draw (u2) node (xi2) {}
			to (-\xfull/2, \ybot * 4.5 / 6)
			to (\xfull/2, \ybot * 4.5 / 6)
			to (u2);
		\node[yshift = .03em, xshift = -.2em] at (\xfull / 3, \ybot * 2/3) {$\xi|_{u_2}$};
		\node[onpath] at (0, \ybot * 1/4) (u7) {};
		\node[onpath] at (-\xfull / 8, \ybot * 1/2) (u8) {};
		\node[xshift = .2em] at (0, 1 * \ybot) (btc2) {$\BTC^{\ell_2}_{w2}(\Phi_{\mathcal{A}})$};
	\end{scope}

	\foreach \u in {ur, u1, u2, u3, u4, u5, u6, u7, u8}{\onpathdeco{\u}}
	\path (ur) edge (u3)
			  (u3) edge (u1)
					  edge (u4)
			  (u1) edge (u5) 
			  (u5) edge (u6)
			  (u2) edge (u7)
			  (u7) edge (u8);
	
	\draw[flow] 
		(u3) edge[out = 30, in = 90] (xw)
		(xw1) to[out = 250, in = 10] ($(u3)!.5!(u2)$) to [out = 190, in = 40] (u6) 
		(u6) edge (btc1)
		(xw2) edge[out = 270, in = 340] (u2) 
		(u2) edge[out = 310, in = 45] (btc2.21);
	
	\coordinate[above = 2em of ur] (xi);
	\draw (xi) --++ (-\xfull, \ybot);
	\draw (xi) --++ (\xfull /2 , \ybot / 2);
	\node[left = 2em of xi, yshift = -1em] {$\xi$};
	
	\end{tikzpicture}
	\caption{}
  \end{subfigure}
\caption{\label{fig:synchronization} Synchronization of the behaviour
  of a wta and of the iteration inherent in the transitive
  closure.} 
\end{figure}
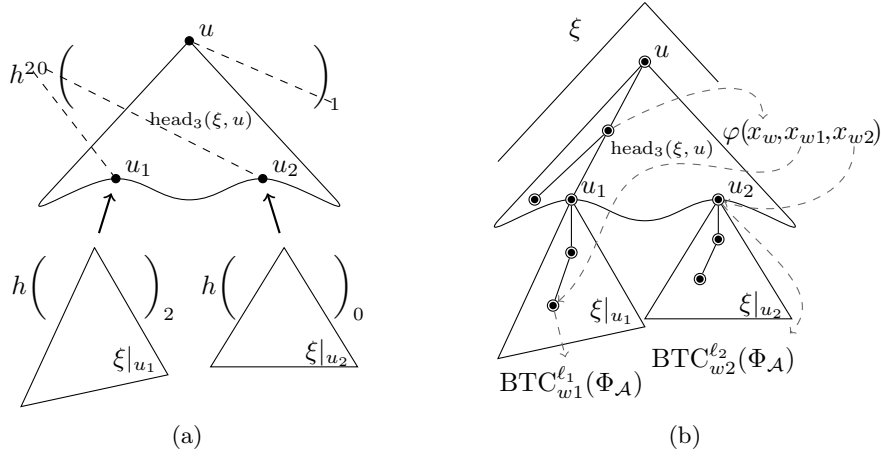

We will represent the states of $\cal A$ by positions of $\xi$. Roughly speaking, the synchronization happens in the way that  $
\seml  \varphi_k(x_w,x_{w1,k})\semr (\xi,v,v_{1,k})$ provides the value
$h^{q_1\ldots q_k}(\head_n(\xi,u))_q$, where $v$ and $v_{1,k}$
are the positions of $\head_n(\xi,u)$ and $\head_n(\xi,u_1),\ldots,\head_n(\xi,u_k)$
which encode $q$ and $q_1,\ldots,q_k$, respectively. Moreover, $\seml
\BTC^{l_i}_{wi}(\Phi_{\cal A})\semr (\xi, v_i)$
provides $h(\xi|_{u_i})_{q_i}$.

\subsection{Decomposition of a Tree into Slices}
\label{sec:slices}

We represent slices as particular trees with variables. For this, we introduce the sets $Z = \{z_1,z_2,z_3,\ldots\}$  and    $Z_k=\{z_1,\ldots,z_k\}$,  $k \in \mathbb{N}$ of variables.
Then we denote by $C_{\Sigma,k}$ the set of all trees $\xi\in T_\Sigma(Z_k)$  such that each $z_i \in Z_k$ occurs exactly once in $\xi$ and the variables occur in the order $z_1,\ldots,z_k$ from left to right. Note that $C_{\Sigma,0}=T_\Sigma$.
For every $k \in \mathbb{N}$, let
$$C_{\Sigma,k}^n = \{\zeta \in C_{\Sigma,k} \mid \forall w \in \pos(\zeta): (|w|<2n) \wedge (\zeta(w) \in Z_k \rightarrow |w|=n) \}.$$
 We note that  $C_{\Sigma,0}^n = \{\xi \in T_{\Sigma} \mid \height(\xi)<2n\}$. Moreover, it should be clear that  there is a $k_0$ (depending also on $\Sigma$) such that $C_{\Sigma,k}^n  =\emptyset$ for every $k > k_0$. We denote the smallest such $k_0$ by  $\max(\Sigma,n)$. It is also clear that  $C_{\Sigma,i}^n \cap C_{\Sigma,j}^n = \emptyset$ for every $i \not= j$.

The next observation is crucial when decomposing a tree into slices.

\begin{ob}\rm \label{ob:unique-breadth} For every $\xi
  \in T_\Sigma$ and $u \in \pos(\xi)$, there is a unique $k \in \mathbb{N}$ and a unique sequence $u_1,\ldots,u_k \in \pos(\xi)$ such that
\begin{itemize}
\item $(\xi[z_1]_{u_1}\ldots [z_k]_{u_k})|_u \in C_{\Sigma,k}^n$ and
\item $\height(\xi|_{u_i}) \ge n$ for every $1 \le i \le k$.
\end{itemize}
\end{ob}

We will denote the tree $(\xi[z_1]_{u_1}\ldots [z_k]_{u_k})|_u$ by
$\head_n(\xi,u)$ and the  sequence $(u_1,\ldots,u_k)$ by
$\cut_n(\xi,u)$. In particular, $\head_n(\xi,u)=\xi|_u$ and
$\cut_n(\xi,u) = (\,)$, i.e., $k =0$, if and only if  $\xi|_u\in
C_{\Sigma,0}^n$. We abbreviate $\head_n(\xi,\varepsilon)$ and
$\cut_n(\xi,\varepsilon)$ by  $\head_n(\xi)$ and $\cut_n(\xi)$, respectively.

The tree $\head_n(\xi,u)$ is the {\em slice of $\xi$ at $u$} and the positions
$u_1,\ldots,u_k$ are {\em cut-positions for $\xi$ and $u$}. By applying Observation \ref{ob:unique-breadth} repeatedly, we obtain a unique decomposition of $\xi$ into slices (cf. Fig. \ref{fig:slices-phi}). Formally, we define the ranked alphabet $C_\Sigma^n$ such that $(C_\Sigma^n)\ui k=C_{\Sigma,k}^n$ for every $k\geq 0$ (recall that $C_\Sigma^n$ is finite). Moreover, we define the mapping $\dec_n: T_\Sigma \rightarrow T_{C_\Sigma^n}$ inductively as follows. For every $\xi \in T_\Sigma$, let
$$\dec_n(\xi) = \head_n(\xi)\Big( \dec_n(\xi|_{u_1}), \ldots, \dec_n(\xi|_{u_k}) \Big)\enspace,$$
where $\cut_n(\xi) = (u_1,\ldots,u_k)$. 

\begin{ob}\label{ob:small}\rm For every $\xi \in T_\Sigma$, $\size(\dec_n(\xi)) = 1$
  if and only if $\height(\xi) < 2n$.
\end{ob}

The following decomposition lemma will be crucial in the simulation of
a wta by means of branching transitive closure.
We note that the lemma can be derived from   \cite[Prop. 18]{mal06a}, which is
proved for bottom-up tree series transducers, i.e. for a generalization of
weighted tree automata. Recall that $S$ is commutative.

\begin{lm}\rm \label{lm:decomp} Let $\xi \in T_\Sigma$, $q \in Q$, and  $\cut_n(\xi) = (u_1,\ldots,u_k)$. Then
$$h(\xi)_q = \sum_{q_1,\ldots,q_k \in Q} h^{q_1\ldots q_k}(\head_n(\xi))_q \cdot \prod_{1 \le i \le k} h(\xi|_{u_i})_{q_i}.$$
\end{lm}
\begin{proof} (Sketch.) We can prove the following, more general statement: for every $k\in \nat$, $\zeta \in C_{\Sigma,k}$ , $\xi_1,\ldots,\xi_k \in T_\Sigma$, and $q\in Q$, we have
$$h(\zeta[\xi_1,\ldots,\xi_k])_q = \sum_{q_1,\ldots,q_k \in Q}
h^{q_1\ldots q_k}(\zeta)_q \cdot \prod_{1 \le i \le k}
h(\xi_i)_{q_i},$$
where $\zeta[\xi_1,\ldots,\xi_k]$ denotes the tree obtained by replacing every occurrence of $z_i$ in $\zeta$ by $\xi_i$ for $1\le i\le k$.
Since the case $k=0$ is trivial, we may assume that $k\in \nat_+$ and proceed by induction on the height of $\zeta$. 
If $\height(\zeta)=0$, then $k=1$ and $\zeta=z_1$, hence the statement  holds again trivially. Now let $\height(\zeta)>0$, i.e., $\zeta=\sigma(\zeta_1,\ldots,\zeta_l)$ for some $l\in \nat_+$, $\sigma \in \Sigma_l$, and $\zeta_1,\ldots,\zeta_l\in T_\Sigma(Z_k)$. By standard arguments, there are $k_1,\ldots,k_l \in \nat$ and there are $\eta_{j} \in C_{\Sigma,{k_j}}$ for $1\leq j \leq l$, such that $k_1+\ldots +k_l=k$ and
$$\zeta[\xi_1,\ldots,\xi_k]=\sigma(\eta_1[\xi_1,\ldots,\xi_{k_1}],\ldots, \eta_l[\xi_{k_1+\ldots+k_{l-1}+1},\ldots,\xi_{k}]).$$
Now we can prove the statement by unfolding $h(\zeta[\xi_1,\ldots,\xi_k])_q$ and organizing the computation appropriately.  In the  first step we apply the weighted transition for $\sigma$. Then  the statement is proved for indexes $j$ with $k_j=0$, while we apply the induction hypothesis on $\height(\zeta)$ for indexes $j$ with $k_j\in\nat_+$. 
\end{proof}

\subsection{The construction of $\Phi_{\cal A}$}
\label{sec:construction-of-Phi}

The formulas $\varphi_k(x,y_{1,k})$ are composed of subformulas that simulate certain properties of ${\cal A}$ (cf. Lemma \ref{lm:phi-semantics}). Let us first establish these subformulas and then assemble $\varphi_k(x,y_{1,k})$. Conceptually, we follow the construction of the corresponding formulas in \cite{bolgasmonzei10} and we borrow several notions from there. However, due to the branching inherent in trees, we have to employ sometimes more sophisticated formulas.

As mentioned we will represent (encode) states of $\cal A$ by positions of the input tree. 
A subtask of $\varphi_k(x,y_{1,k})$ is to find out, for a position
$v$, the base position of $v$ and the state encoded by $v$. Next we elaborate the corresponding formulas.

\paragraph{Identifying Base Positions and Coded States.}

Let $\xi \in T_\Sigma$ and $v \in \pos(\xi)$. In Section~\ref{bTC-section} we have defined the base position $\langle v\rangle$ of $v$. We can use the macro $y=\langle x\rangle_n$ in $\bFOmod$ to identify the base position in the sense that:
\begin{equation}
\seml y=\langle x\rangle_n \semr(\xi,v,u) =\\
\left\{
\begin{array}{ll}
1 & \hbox{ if } u = \langle v \rangle\\
0 & \hbox{ otherwise} \enspace.
\end{array}
\right. \label{equ:base-node}
\end{equation}
Then the {\em state encoded by $v$} is the number $|v| - |\langle v \rangle|$. This can be turned into a formula by finite disjunction. Let us denote this number by $[v]$. Due to the definition of $\langle v \rangle$ we have that $[v] \in \{0,\ldots,n-1\}$. Note that $\langle v \rangle \in \pos(\xi)$ and $[v] \in Q$.

For reasons detailed later, we would like the base position $\langle v\rangle$ to coincide with a cut-position of $\xi$. But then, due to the branching inherent in $\xi$, the state $[v]$ may be represented by any node $v'$ satisfying that $\langle v'\rangle=\langle v\rangle$
and $[v']=[v]$. We will avoid this by forcing the assignment to choose a $v$ which is on the leftmost path from $\langle v \rangle$, and this leftmost path must have at least length $n-1$ (in order to be able to encode each of the $n$ states). Thus we define the following macros to identify states:
\label{p:onlmp}

\begin{itemize}
\item $\onlmp_{n-1}(x,y) := \underline{\exists} y_{1,n}. (x=y_1) \wedge
  \formlmp(y_{1,n}) \wedge \underline{\bigvee}_{1 \le i
    \le n} (y_i=y)$\\(there is a path of length $n-1$ starting from $x$, and $y$ is a position of the leftmost such path), 

\item $\formlmp(y_{1,n}) :=
  \formpath(y_{1,n}) \wedge$ \\
\hspace*{27mm} $ \forall z.\bigg[\bigg((y_1 \le_{n-1} z) \wedge (z \not= y_n)\bigg)
\stackrel{+}{\rightarrow} \sibl_{n-1}(y_{n},z) \bigg]$ \\
(the positions $y_1,\ldots, y_n$ form the leftmost path of length $n-1$),

\item $\formpath(y_{0,n}) := \underline{\bigvee}_{w \in \{1,\ldots,\maxrk(\Sigma)\}^n}
\formpath_w(y_{0,n})$\\
($y_0,\ldots,y_n$ form a path).
\end{itemize}

\paragraph{Identifying the Cut-Positions.}

Due to Observation \ref{ob:unique-breadth}, any position $u$ uniquely determines the sequence $\cut_n(\xi,u)$ of cut-positions. The next subtask of $\varphi_k(x,y_{1,k})$ is to identify this sequence. For this we employ the macro $\formcut_{n,k}(x,y_{1,k})$ with $k \ge 0$ such that, for every $u,u_1,\ldots,u_k \in \pos(\xi)$: 
\begin{equation}
\seml \formcut_{n,k}(x, y_{1,k}) \semr(\xi,u,u_{1,k}) =\\
\left\{
\begin{array}{ll}
1 & \hbox{ if } \cut_n(\xi,u)=(u_1,\ldots,u_k) \hbox{ and }\\
0 & \hbox{ otherwise.}
\end{array}
\right. \label{equ:form-cut}
\end{equation}
We define
\begin{align*}\formcut_{n,k}(x,y_{1,k}):= & \left(\bigwedge_{i=1}^{k}
    (x \le_n y_i)\wedge (\height(y_i) \geq n)\right) \wedge
  \left(\bigwedge_{i=1}^{k-1} \sibl_n(y_i, y_{i+1})\right) \wedge \\ &
  \left(\forall z. ((x \le_n z) \wedge (\height(z) \geq n))  \stackrel{+}{\rightarrow}\big(\underline{\bigvee}_{i=1}^k z=y_i\big)\right),
\end{align*}
where we have used the following auxiliary macros:
\begin{itemize}
\item $(\height(x) \ge n) := \underline{\exists} z. (x \le_n z)$
\end{itemize} \label{p:height}
Taking the definition of $\cut_n(\xi,u)$ into account, it is not
difficult to see that our macro satisfies (\ref{equ:form-cut}). 
In particular, $\formcut_{n,0}^\xi(u)$ holds if and only if $\xi|_u \in C^n_{\Sigma,0}$.

\paragraph{Identifying the Head.}

For every $u\in \pos(\xi)$ with $\cut_n(\xi,u) = (u_1,\ldots,u_k)$, we
can identify the piece of $\xi$ which starts at $u$ and ends in
$(u_1,\ldots,u_k)$, which is $\head_n(\xi,u)$. More precisely, for every  $k\in \nat$ and $\zeta \in C_{\Sigma,k}^n$ we define the macro $\chk_\zeta(x,y_{1,k})$ such that for every $\xi \in T_\Sigma$, $u,u_1,\ldots,u_k \in \pos(\xi)$:
\begin{equation}
\seml \chk_\zeta(x,y_{1,k}) \semr (\xi,u,u_{1,k})
=
\left\{
\begin{array}{ll}
1 & \hbox{ if }\zeta = (\xi[z_1]_{u_1} \ldots [z_k]_{u_k})|_u\\
0 & \hbox{ otherwise.}
\end{array} \label{equ:check}
\right.
\end{equation}
Hence in case $k=0$ we have
\[
\seml \chk_\zeta(x) \semr (\xi,u)
=
\left\{
\begin{array}{ll}
1 & \hbox{ if }\zeta = \xi|_u\\
0 & \hbox{ otherwise.}
\end{array}
\right.
\]
The definition of the macro is as follows: \label{p:check}
\begin{align*}\chk_\zeta & (x,y_{1,k}) := \\ & \bigwedge\limits_{w \in
  \pos(\zeta) \setminus \pos_{Z_k}(\zeta)} (\underline{\exists}
y. \; (x \le_w y) \wedge \mathrm{label}_{\zeta(w)}(y))
\wedge \bigwedge\limits_{1 \le i \le k} (x\le_{\pos_{z_i}(\zeta)} y_i) \enspace.
\end{align*}
In case $k=0$ we have
\begin{align*}\chk_\zeta & (x) =  \bigwedge\limits_{w \in
  \pos(\zeta) } (\underline{\exists}
y. \; (x \le_w y) \wedge \mathrm{label}_{\zeta(w)}(y))
\end{align*}
It is easy to observe that (\ref{equ:check}) is satisfied.

\paragraph{Construction of $\Phi_{\cal A}$.}

Now we define the family $\Phi_{\cal A} = (\varphi_k(x,y_{1,k}) \mid 0
\le k \le \max(\Sigma,n))$ of MSO-formulas where 
\[
\varphi_0(x) :=  \bigvee\limits_{0 \le q \le
  n-1}\bigvee\limits_{\zeta \in C_{\Sigma,0}^n} \left( \,\underline{\exists} z.\theta_{q,\zeta}(x,z)\right) \wedge h(\zeta)_q
\]
with 
\begin{eqnarray*}
 \theta_{q,\zeta}(x,z) :=  (z=\langle x\rangle_n) 
\wedge  \formcut_{n,0}(z) \wedge  (z \le_q x) \wedge
\chk_\zeta(z) \enspace,
\end{eqnarray*}
and for every $1 \le k \le \max(\Sigma,n)$ 
\[
\varphi_k(x,y_{1,k}) :=  
\bigvee\limits_{0 \le q_1,\ldots,q_k,q \le n-1}
\bigvee\limits_{\zeta \in C_{\Sigma,k}^n} 
\left( \,\underline{\exists} z,z_{1,k}.\theta_{q,q_{1,k},\zeta}(x,y_{1,k},z,z_{1,k})\right) \wedge
h^{q_1\ldots q_k}(\zeta)_q 
\]
with 
\begin{align*}
\theta_{q,q_{1,k},\zeta}(x,y_{1,k},z,z_{1,k}) :=
 &(z=\langle x\rangle_n)   
\wedge \formcut_{n,k}(z,z_{1,k})\wedge \bigwedge_{1 \le i \le k} \onlmp_{n-1}(z_i,y_i) \wedge \\
& (z \le_q x) \wedge \bigwedge_{1 \le i \le k}  (z_i \le_{q_i} y_i)\wedge  
\chk_\zeta(z,z_{1,k})\enspace.
\end{align*}
Note that $\varphi_k$ is a weighted disjunction and not the Boolean one.

\begin{lm}\rm \label{ob:Phi-bFOmod} $\Phi_{\cal A}$ is a  $\max(\Sigma,n)$-family of  $n$-progressing formulas in  $(\bFOmod)_\mathrm{step}$.
\end{lm}
\begin{proof} First, it is easy to check that each formula $\varphi_k$  is in  $(\bFOmod)_\mathrm{step}$. 

Second, we show that the implication (\ref{equ:progress}) holds. For this, let us assume that $\seml\varphi_{k}(x,y_{1,k})\semr(\xi,v,v_{1,k}) \neq 0$. Due to the definition of $\varphi_k$ there are positions $u, u_{1,k} \in \pos(\xi)$ such that 
\begin{itemize}
\item $u = \langle v \rangle$,
\item $\formcut_{n,k}^\xi(u,u_{1,k})$ and $\onlmp_{n-1}^\xi(u_i,v_i)$ hold
\end{itemize}
Since $\formcut_{n,k}^\xi(u,u_{1,k})$ holds, also $(u \le_n u_i)$ holds, and $\onlmp_{n-1}^\xi(u_i,v_i)$ implies that $|v_i| < |u_i| + n$. Thus  $u_i = \langle v_i \rangle$ for every $1 \le i \le k$. Moreover, $\formcut_{n,k}^\xi(u,u_{1,k})$ implies $\sibl_n^\xi(u_i,u_{i+1})$ for every $1 \le i \le k-1$. This means that $\psi_{k}^\xi(v,v_{1,k})$ holds. 
\end{proof}


\subsection{Proof of Theorem \ref{BbTC-theo}}

Now we will prove Theorem \ref{BbTC-theo}.
We split the proof into three steps. In the first step we determine
the semantics of the formula $\varphi_k(x,y_{1,k})$. We prepare this
by the following technical lemma.

\begin{lm}\rm \label{plus-lemma} For every $\xi\in T_\Sigma$, $0 \le k \le \max(\Sigma,n)$, $v,v_1,\ldots,v_k \in \pos(\xi)$,  $0 \le q_1,\ldots,q_k,q \le  n-1$, and $\zeta \in C_{\Sigma,k}^n$ we have
\begin{quote}
$(\underline{\exists}z,z_{1,k}.\theta_{q,q_{1,k},\zeta})^\xi(v,v_{1,k})$ holds $\iff$
$ \theta^\xi_{q,q_{1,k},\zeta}(v,v_{1,k},\langle v \rangle,\langle v\rangle_{1,k})$ holds,
\end{quote}
 where $\langle v\rangle_{1,k}$ abbreviates the sequence $\langle
v_1 \rangle, \ldots, \langle v_k \rangle$.
\end{lm}
\begin{proof} The direction $\Leftarrow$ holds by definition. 
To show the direction $\Rightarrow$, assume that there are $u,u_1,\ldots,u_k \in \pos(\xi)$ such that $\theta^\xi_{q,q_{1,k},\zeta}(v,v_{1,k},u,u_{1,k})$ holds. Then, in particular, we have that 
\begin{enumerate}
\item [(a)] $(u=\langle v\rangle)^\xi$, 
\item [(b)] $\formcut_{n,k}^\xi(u,u_{1,k})$, and
\item [(c)] $(u_i \le_{q_i} v_i)^\xi$ holds for every $1\leq i \leq k$.
\end{enumerate}
Hence $u=\langle v \rangle$ by (a). By (b), we have $\formcut_{n,k}^\xi(\langle v \rangle,u_{1,k})$.  This latter, Condition (c), and the fact that $0\leq q_i \leq n-1$ for every $1\leq i\leq k$ imply that $u_i=\langle v_i \rangle$ for every $1\leq i\leq k$.
\end{proof}

Now we are able to characterize  the semantics of  $\varphi_k(x,y_{1,k})$.

\begin{lm}\rm \label{lm:phi-semantics} For every $\xi\in T_\Sigma$, $0 \le k \le \max(\Sigma,n)$,   and $v,v_1,\ldots,v_k \in \pos(\xi)$, we have
\[
\seml \varphi_k(x,y_{1,k})\semr (\xi, v, v_{1,k}) =
\left\{
\begin{array}{ll}
 h^{[v_1]\ldots [v_k]}(\head_n(\xi,\langle v \rangle)_{[v]} 
& \hbox{if }\formcut_{n,k}^\xi(\langle v \rangle,\langle v\rangle_{1,k}) \hbox{ and } \\ 
& \onlmp^\xi_{n-1}(\langle v_i
\rangle,v_i) \hbox{ hold }\\
& \hbox{for every } 1 \le i \le k\\[2mm]
0 & \hbox{otherwise } \enspace.
\end{array}
\right.
\]
\end{lm}
\begin{proof}  
\underline{Case 1:} $\formcut_{n,k}^\xi(\langle v \rangle,\langle v\rangle_{1,k})$ and  $\onlmp^\xi_{n-1}(\langle v_i
\rangle,v_i)$ hold for every $1 \le i \le k$. Then
$\chk^\xi_\zeta(\langle v \rangle,\langle v
\rangle_{1,k})$ holds for  $\zeta =
\head_n(\xi,\langle v \rangle)$ (due to Equation (\ref{equ:check})). Moreover,
\begin{quote}
$(\langle v\rangle \le_q v)^\xi$ holds iff  $q = [v]$
and 
$(\langle v_i\rangle \le_{q_i} v_i)^\xi$ holds  iff $q_i = [v_i]$.
\end{quote}

\noindent Then  $\theta^\xi_{[v],[v]_{1,k},\zeta}(v,v_{1,k},\langle v\rangle,\langle v\rangle_{1,k})$
holds and thus, by Lemma \ref{plus-lemma},
$(\underline{\exists}z,z_{1,k}.\theta_{[v],[v]_{1,k},\zeta})^\xi(v,v_{1,k})$ holds, where $[v]_{1,k}$ abbreviates the sequence $[v_1],\ldots,[v_k]$. 
Altogether this means that $\seml \varphi_k(x,y_{1,k})\semr (\xi, v, v_{1,k}) =
 h^{[v_1]\ldots [v_k]}(\head_n(\xi,\langle v \rangle))_{[v]}$.

\underline{Case 2:}  $\formcut_{n,k}^\xi(\langle v \rangle,\langle v\rangle_{1,k})$ does not hold or $\onlmp^\xi_{n-1}(\langle v_i
\rangle,v_i)$ does not hold for some $1 \le i \le k$. Then
for every $0 \le q_{1,k},q \le n-1$ and $\zeta \in C_{\Sigma,k}^n$, the property
  $\theta^\xi_{q,q_{1,k},\zeta}(v,v_{1,k},\langle v\rangle, \langle v\rangle_{1,k})$ does not hold and thus, by Lemma \ref{plus-lemma}, 
$(\underline{\exists}z,z_{1,k}.\theta_{q,q_{1,k},\zeta})^\xi(v,v_{1,k})$ does not hold. Hence
   $\seml
  \varphi_k(x,y_{1,k})\semr (\xi, v, v_{1,k}) = 0$.
\end{proof}

In the second step, we prove that in the disjunction (on $l$) which defines $\seml \BTC(\Phi_{\cal A}) \semr(\xi,\varepsilon)$ only one member may differ from 0. In the following we abbreviate  $\Phi_{\cal A}$ by $\Phi$.

\begin{lm} \rm \label{lm:level}
For every  $l \in \nat_+$, $w \in \nat_+^*$, $\xi \in T_\Sigma$, and $v \in \pos(\xi)$, if $l
\not= \size(\dec_n(\xi|_{\langle v \rangle}))$, then $\seml \BTC^{l}_w(\Phi)\semr(\xi,v) =0$. Hence $$\seml \BTC(\Phi)\semr(\xi,v) = \seml \BTC^{\size(\dec_n(\xi|_{\langle v \rangle}))}_w(\Phi)\semr(\xi,v).$$
\end{lm}
\begin{proof} We prove the statement by induction on $l$.

\underline{$l=1$}: By our assumption $\size(\dec_n(\xi|_{\langle v
  \rangle}))> 1$. Then $\seml \formcut_{n,0}(x)\semr (\xi, v)=0$ and
thus we have $\seml \BTC^{1}_w(\Phi)\semr(\xi,v) =0$

\underline{$l\Rightarrow l+1$}: Let us assume that $l+1
\not= \size(\dec_n(\xi|_{\langle v \rangle}))$ and that for every $l'\le l$, $w' \in \nat_+^*$, and $v' \in \pos(\xi)$, if $l'\not= \size(\dec_n(\xi|_{\langle v' \rangle}))$, then $\seml
\BTC^{l'}_{w'}(\Phi)\semr(\xi,v') =0$. We prove by contradiction. Therefore, we assume that $\seml
\BTC^{l+1}_w(\Phi)\semr(\xi,v) \ne 0$. This latter, by definition, means that there are $k\geq 1$ and $v_{1},\ldots,v_k \in \pos(\xi)$ such that
\begin{enumerate}
\item[(a)] $\seml \varphi_{k}(x_w,x_{w1,k}) \semr (\xi,v, v_{1,k}) \ne 0$, and 
\item[(b)] $\seml \bigvee\limits_{\substack{l_1,...,l_k \in \nat_+\\l_1+...+l_k = l}}
\bigwedge\limits_{1 \le i \le k} 
 \BTC^{l_i}_{wi}(\Phi)\semr (\xi,v_{1,k}) \ne 0$.
\end{enumerate}
Condition (a) implies that $\psi_{k}^\xi(v,v_{1,k})$ holds. 
By condition (a) and Lemma \ref{lm:phi-semantics}, we obtain that $\formcut_{n,k}^\xi(\langle v\rangle,\langle v\rangle_{1,k})$ holds, which means that $\cut_n(\xi,\langle v \rangle) = (\langle v_1\rangle,\ldots,\langle v_k\rangle)$ (Equation \ref{equ:form-cut}).  Thus
$$\size(\dec_n(\xi|_{\langle v \rangle}))=1+\sum_{i=1}^k \size(\dec_n(\xi|_{\langle v_i \rangle})).$$
Moreover, condition (b) means that there are $l_1,...,l_k \in \nat_+$ with $l_1+...+l_k = l$ such that, for every $1\leq i\leq k$, we have
$\seml \BTC^{l_i}_{wi}(\Phi)\semr (\xi, v_i) \neq 0$.
On the other hand, by our assumption, there is a $1\leq j\leq k$ such that $l_j \neq \size(\dec_n(\xi|_{\langle v_j \rangle}))$.
For this $j$, by the induction hypothesis, we have $\seml \BTC^{l_j}_{wj}(\Phi)\semr (\xi, v_j) = 0$, which is a contradiction.
Hence $\seml \BTC^{l+1}_w(\Phi)\semr(\xi,v) = 0$.
\end{proof}

In the third step we prove that in the disjunction (on $k$) which defines $\BTC^{l+1}_w(\Phi)$  only one member may differ from 0.

\begin{lm}\rm  \label{lm:rank} Let $\xi \in T_\Sigma$, $v \in
  \pos(\xi)$ with $\cut_n(\xi,\langle v \rangle) = (u_1,\ldots,u_k)$ for some $k \in
  \mathbb{N}$, and  $w \in \nat_+^*$. Then  
\[\seml \BTC^{l+1}_w(\Phi)\semr(\xi,v) =
\seml \exists  x_{w1,k}. \; \varphi_k(x_w,x_{w1,k}) \wedge
\bigvee\limits_{\substack{l_1,...,l_k \in \nat_+\\l_1+...+l_k = l}}
\bigwedge\limits_{1 \le i \le k}
 \BTC^{l_i}_{wi}(\Phi)\semr (\xi, v)\]
for every $l \in \nat$.
\end{lm}
\begin{proof} First we show by contradiction that, for every $k' \in \mathbb{N}$ with $k' \not= k$ and
  $v_1,\ldots,v_{k'} \in \pos(\xi)$, we have that $\seml
  \varphi_{k'}(x,y_{1,k'})\semr (\xi,v,v_{1,k'}) = 0$.
Assume that there are $k' \,(\not= k)$ and 
  $v_1,\ldots,v_{k'} \in \pos(\xi)$  such that $\seml
  \varphi_{k'}(x,y_{1,k'})\semr (\xi,v,v_{1,k'}) \not= 0$. Then, by Lemma \ref{lm:phi-semantics}, we have $\formcut_{n,k'}^\xi(\langle v\rangle,\langle v\rangle_{1,k'})$ holds, i.e., $\cut_n(\xi,\langle v \rangle) = (\langle v_1\rangle,\ldots,\langle v_{k'}\rangle)$ (by Equation \ref{equ:form-cut}). 
 But this contradicts the fact that the breadth of
$\cut_n(\xi,\langle v \rangle)$ is $k$ and  the uniqueness of the breadth of a cut
(cf. Observation \ref{ob:unique-breadth}). 
Then we can calculate as follows:

\

\hspace*{-6mm}\begin{tabular}{ll}
 & $\seml \BTC^{l+1}_w(\Phi) \semr(\xi,v)$\\[2mm]

= & $\seml\bigvee\limits_{0 \le k' \le m} \exists
  x_{w1,k'}.  \varphi_{k'}(x_w,x_{w1,k'}) \wedge
\bigvee\limits_{\substack{l_1,...,l_{k'} \in \nat_+\\l_1+...+l_{k'} = l}}
\bigwedge\limits_{1 \le i \le k'}
 \BTC^{l_i}_{wi}(\Phi)\semr(\xi,v)$\\[6mm]

= & $\sum\limits_{0 \le k' \le m}\seml \exists
  x_{w1,k'}.  \varphi_{k'}(x_w,x_{w1,k'}) \wedge
\bigvee\limits_{\substack{l_1,...,l_{k'} \in \nat_+\\l_1+...+l_{k'} = l}}
\bigwedge\limits_{1 \le i \le k'}
 \BTC^{l_i}_{wi}(\Phi)\semr(\xi,v)$\\[6mm]

= & $\seml \exists  x_{w1,k}.   \varphi_k(x_w,x_{w1,k}) \wedge
\bigvee\limits_{\substack{l_1,...,l_k \in \nat_+\\l_1+...+l_k = l}}
\bigwedge\limits_{1 \le i \le k}
 \BTC^{l_i}_{wi}(\Phi)\semr (\xi, v)$\\[2mm]

& (since $\seml
  \varphi_{k'}(x_w,x_{w1,k'})\semr (\xi,v,v_{1,k'}) = 0$ for every $k'\ne k$ and $v_1,\ldots,v_{k'}\in\pos(\xi)$\\
& by Lemma \ref{lm:phi-semantics})\enspace.
\end{tabular} 

\

\noindent This proves the statement.
\end{proof}

\

\begin{proof} of Theorem \ref{BbTC-theo}. Let $\xi \in T_\Sigma$. 

  \underline{Case 1:} $\height(\xi)<n$. Then 
\[
\begin{array}{cl}
  & \seml \BTC(\Phi)\semr( \xi, \varepsilon)\\[1mm]

= & \seml \BTC^{ \size(\dec_n(\xi))}_\varepsilon(\Phi)\semr( \xi,
\varepsilon) \hspace{5mm} \hbox{(by Lemma \ref{lm:level})}\\[2mm]

= & \seml \BTC^{1}_\varepsilon(\Phi)\semr( \xi, \varepsilon) \hspace{5mm}
\hbox{(because $\height(\xi)<n$ and by Obs. \ref{ob:small})} \\[2mm]

= & \seml \varphi_0(x_\varepsilon)\semr( \xi, \varepsilon)  \hspace{5mm}
\hbox{(by definition of $\BTC_\varepsilon^{1}(\Phi)$)} \\[2mm]

= &  h(\xi)_{0} \hspace{5mm}
\hbox{(by Lemma \ref{lm:phi-semantics} and the fact that
  $[\varepsilon] = 0$)}\\[2mm]

= & r_{\cal A}(\xi) \enspace.
\end{array}
\]

\underline{Case 2:} $\height(\xi)\geq n$. We consider
the following statement:

\label{Statement1}
\begin{quote}\underline{Statement 1.} For every  $l \ge 1$, $w\in \nat_+^*$, and  $v \in \pos(\xi)$, \\
if $l = \size(\dec_n(\xi|_{\langle v \rangle}))$ and
$\onlmp^\xi_{n-1}(\langle v\rangle,v)$ holds, then $\seml \BTC^{l}_w(\Phi)\semr( \xi, v) = h(\xi|_{\langle v \rangle})_{[v]}.$
\end{quote}

If Statement 1 holds, then we obtain
$$\seml \BTC(\Phi)\semr( \xi, \varepsilon)
=  \seml \BTC^{\size(\dec_n(\xi))}_\varepsilon(\Phi)\semr( \xi, \varepsilon) = h(\xi)_{[\varepsilon]} =  h(\xi)_0 =  r_{\cal A}(\xi),$$
where the first and the second equalities are justified by  Lemma \ref{lm:level} and Statement 1, respectively.

Finally, we prove Statement 1 by induction on $l$. 

\underline{$l=1$}: We have
\begin{align*}
\seml \BTC^{1}_w(\Phi)\semr( \xi, v)& =  \seml \varphi_0(x_w)\semr( \xi, v)  \hspace{5mm}
\hbox{(by the definition of $\BTC_w^{1}(\Phi)$)} \\[2mm]
& = h(\xi|_{\langle v\rangle})_{[v]} 
\hbox{\hspace{15mm}(by Lemma \ref{lm:phi-semantics})} \enspace.
\end{align*}

\underline{$l\Rightarrow l+1$}:  We assume
that  $l + 1 = \size(\dec_n(\xi|_{\langle v \rangle}))$ and
that  Statement 1 holds for every  $1\le l' \le l$. We denote the cut-positions
below $\langle v \rangle$ by $u_i$, i.e., $\cut_n(\xi,\langle v \rangle) =
(u_1,\ldots,u_k)$ for some  $k \ge 1$. Then we can calculate as follows.

\

\begin{tabular}{ll}

& $\seml \BTC^{l+1}_w(\Phi)\semr( \xi, v)$\\[5mm]

= & $\seml \exists x_{w1,k}. \; 
  \varphi_k(x_w,x_{w1,k})  \wedge
\bigvee\limits_{\substack{l_1,...,l_k \in \nat_+\\l_1+...+l_k = l}}
\bigwedge\limits_{1 \le i \le k}
 \BTC^{l_i}_{wi}(\Phi)\semr (\xi, v)$\\[-3mm]
& (by Lemma \ref{lm:rank})\\[6mm]

= & $\sum\limits_{v_1,\ldots,v_k \in \pos(\xi)}
\seml   \varphi_k(x_w,x_{w1,k}) \; \wedge \; 
\bigvee\limits_{\substack{l_1,...,l_k \in \nat_+\\l_1+...+l_k = l}}
\bigwedge\limits_{1 \le i \le k}
 \BTC^{l_i}_{wi}(\Phi)\semr (\xi, v,v_{1,k})$\\[12mm]

= & $\sum\limits_{v_1,\ldots,v_k \in \pos(\xi)}
\bigg[\seml   \varphi_k(x_w,x_{w1,k}) \semr  (\xi, v,v_{1,k}) \cdot 
\sum\limits_{\substack{l_1,...,l_k \in \nat_+\\l_1+...+l_k = l}}
\prod\limits_{1 \le i \le k}
 \seml \BTC^{l_i}_{wi}(\Phi)\semr (\xi,v_i)\bigg]$\\[9mm]

= & $\sum\limits_{\substack{v_1,\ldots,v_k \in \pos(\xi):\\u_1=\langle v_1\rangle, \ldots, u_k=\langle v_k\rangle,\\
\onlmp^\xi_{n-1}(\langle v_i \rangle,v_i) }}
\bigg[h^{[v_1]\ldots [v_k]}(\head_n(\xi,\langle v \rangle))_{[v]} \cdot$\\[-7mm]
& \hfill $\sum\limits_{\substack{l_1,...,l_k \in \nat_+\\l_1+...+l_k = l}}
\prod\limits_{1 \le i \le k}
 \seml \BTC^{l_i}_{wi}(\Phi)\semr (\xi,v_i)\bigg]$\\[5mm]
&(by Lemma \ref{lm:phi-semantics})\\[8mm]

= & $\sum\limits_{\substack{v_1,\ldots,v_k \in \pos(\xi):\\u_1=\langle v_1\rangle, \ldots, u_k=\langle v_k\rangle,\\
\onlmp^\xi_{n-1}(\langle v_i \rangle,v_i)}}
\bigg[h^{[v_1]\ldots [v_k]}(\head_n(\xi,\langle v \rangle))_{[v]} \cdot$\\[-5mm] 
&\hfill  $\prod\limits_{1 \le i \le k}
 \seml \BTC^{\size(\dec_n(\xi|_{\langle v_i\rangle}))}_{wi}(\Phi)\semr (\xi,v_i)\bigg]$\\[3mm]

& (by Lemma \ref{lm:level})\\
\end{tabular}

\begin{tabular}{ll}

= & $\sum\limits_{
\substack{v_1,\ldots,v_k \in \pos(\xi):\\
u_1=\langle v_1\rangle, \ldots, u_k=\langle v_k\rangle,\\
\onlmp^\xi_{n-1}(\langle v_i \rangle,v_i)
}}
\bigg[h^{[v_1]\ldots [v_k]}(\head_n(\xi,\langle v \rangle))_{[v]}
 \cdot 
\prod\limits_{1 \le i \le k}
h(\xi|_{\langle v_i \rangle})_{[v_i]}
\bigg]$\\[10mm]

& (by I.H.) \\[4mm]

= & $\sum\limits_{q_1,\ldots,q_k \in Q}
\bigg[h^{q_1\ldots q_k}(\head_n(\xi,\langle v \rangle))_{[v]} \cdot \prod\limits_{1
  \le i \le k} h(\xi|_{u_i})_{q_i}\bigg]$\\[5mm]

= & $h(\xi|_{\langle v \rangle})_{[v]}$\\

& (by Lemma \ref{lm:decomp}) \enspace.
\end{tabular}

\

\noindent The last but one step is justified by the fact that there is a one-to-one correspondence between the two index sets. In fact, it is easy to see that, for every $1\leq i\leq k$, the set $\{ v\in \pos(\xi) \mid u_i =\langle v\rangle \text{ and } \onlmp_{n-1}(u_i,v) \}$
has exactly $n$ elements.
\end{proof}

\section{From Branching Transitive Closure to $\exists\forall(\bMSO_\mathrm{step})$}
\label{sect:BTC-EA}

In this section let  ${\cal L}$ be a fragment of BMSO which contains
BFO+mod and which is closed under conjunction and the quantification
$\underline{\exists} x$. Our goal is to prove the following theorem.

\begin{theo}\rm \label{EA-theo} Let $m \in \mathbb{N}$ and $n \in
  \mathbb{N}_+$. For every $m$-family $\Phi = (\varphi_k(x,y_{1,k})
  \mid 0 \le k \le m)$ of $n$-progressing formulas in ${\cal
    L}_{\mathrm{step}}$ there is an $\exists\forall({\cal
    L}_{\mathrm{step}})$-formula $\Theta$ such that  $\seml \BTC(\Phi)\semr (\xi,\varepsilon) = \seml \Theta \semr(\xi)$ for every $\xi \in T_\Sigma$. 
\end{theo}

\subsection{Construction of $\Theta$}

Clearly, $\Theta$ should have the form $\exists X. \forall x. \theta(X,x)$ for some ${\cal L}$-step formula $\theta(X,x)$. 
First  we introduce the macro $\sel(X,z)$ which is a conjunction of two formulas. If $X$ and $z$
are assigned the set $J$ of positions and the position $v$,
respectively, then the first conjunct expresses that for every node
$u\in J$ (except if the base position of $u$ is $v$), there is another node $u'\in J$ such that 
$\langle u \rangle =\langle u' \rangle w$ for some string $w$ of length $n$.
The second conjunct expresses that there are no two different selected
nodes $u$ and $u'$ in $J$ such that $\langle u\rangle =\langle
u'\rangle$. These two properties of $J$ and $v$ assure that if $sel^\xi(J,v)$ holds, then the nodes in $J$  are situated as, e.g., the solid nodes in Fig. \ref{fig:unfolding-fig}.

The exact definition is
\begin{align*}\sel(X,z):= & \forall x. \bigg[(x\in X) \pimplies \bigg(
  z=\langle x \rangle_n \vee \underline{\exists} y. \big(y\in X \wedge \psi_{1}(y,x)\big)\bigg)\bigg] \wedge \\
& \forall x,y. \bigg[ \bigg( (x\in X) \wedge (y \in X) \wedge \underline{\exists} z'.( z' =\langle x \rangle_n \wedge z' =\langle y \rangle_n)\bigg) \pimplies (x = y) \bigg].
\end{align*}

Then we define
$$\Theta = \exists X.\forall x. \theta(X,x)$$
where
\begin{itemize}
\item $\theta(X,x) :=  \; \sel(X,\varepsilon) \wedge (\varepsilon \in X) \wedge  \bigg( (x\in X) \pimplies 
\bigvee_{k=0}^m\exists y_{1,k}. \; \theta_k(X,x,y_{1,k})\bigg)$

\item $\theta_k(X,x,y_{1,k}) :=  \;  \varphi_k(x,y_{1,k}) \wedge (y_{1,k}\in X) \wedge$\\ 
\hspace*{50mm} $\wedge \forall y. \big ((y\in X) \wedge \psi_{1}(x,y) \pimplies \underline{\bigvee}_{i=1}^k (y = y_i)\big)$

\item $\sel(X,\varepsilon) = \underline{\exists} x. \mathrm{root}(x)
\wedge \sel(X,x)$,  and 

\item  $(\varepsilon \in X) := \underline{\exists} x. \mathrm{root}(x) \wedge
(x \in X)$.
\end{itemize}

\subsection{$\Theta$ is equivalent to a $\exists\forall({\cal L}_{\mathrm{step}})$-formula}

First we prove the following technical lemma on the subformula $\theta_k(X,x,y_{1,k})$. 
\begin{lm}\rm\label{technical-lemma} For every $\xi\in T_\Sigma$ $J\subseteq \pos(\xi)$, and $u\in \pos(\xi)$, there is at most one $k\geq 0$ and sequence $u_{1,k}\in \pos(\xi)$ such that
\[
\seml \theta_k(X,x,y_{1,k})\semr (\xi,J,u,u_{1,k})\neq 0.
\]
\end{lm}
\begin{proof} We prove by contradiction. Let us assume that there are $k, l \geq 0$ and  sequences $u_{1,k}, v_{1,l}\in \pos(\xi)$ such that $\seml \theta_k(X,x,y_{1,k})\semr (\xi,J,u,u_{1,k})\neq 0$ and $\seml \theta_l(X,x,y_{1,l})\semr (\xi,J,u,v_{1,l})\neq 0$. Assume also that $v_j \not\in \{u_1,\ldots,u_k \}$ for some $1\le j\le l$. 

Since $\seml \varphi_k(x,y_{1,k})\semr (\xi,u,u_{1,k})\neq 0$, by the implication (\ref{equ:progress}),
we have $\psi_{k}^\xi(u,u_{1,k})$. Analogously, we have $\psi_{k}^\xi(u,v_{1,l})$, which implies $\psi_{1}^\xi(u,v_j)$. We also have $u_{1,k}\in J$ and  $v_{1,l}\in J$.

We also have 
\[\seml \forall y. \big ((y\in X) \wedge \psi_{1}(x,y) \pimplies \underline{\bigvee}_{i=1}^k (y = y_i)\big)\semr (\xi,u,u_{1,k}) =1,\]
hence
\[\seml y\in X \wedge \psi_{1}(x,y) \pimplies \underline{\bigvee}_{i=1}^k (y = y_i)\semr (\xi,u,v_j,u_{1,k}) =1.\]
However, the latter implies that $v_j=u_i$ for some $1\le i\le k$, contradiction our assumption. This means $ \{v_1,\ldots,v_l \}= \{u_1,\ldots,u_k \}$. Finally, we note that the order $u_1,\ldots,u_k$ is uniquely determined by the $\sibl_n$ relation, which is a part of $\psi_{k}$.
\end{proof}

\begin{lm}\rm\label{BMSO-step-lemma} $\Theta$ is equivalent to a $\exists\forall({\cal L}_{\mathrm{step}})$-formula.
\end{lm}
\begin{proof}We show that the formula $\theta(X,x)$
is equivalent to an ${\cal L}$-step formula. Let us apply $\varphi \pimplies \psi := \neg \varphi \vee (\varphi \wedge \psi)$ to the first occurrence of $\pimplies$. Then, since   $\sel(X,z)$ is in BFO+mod, it suffices to show that
\[\exists y_{1,k}.\big[ \varphi_k(x,y_{1,k}) \wedge  (y_{1,k}\in X) \wedge 
 \forall y. \big ((y\in X) \wedge \psi_{1}(x,y) \pimplies\underline{\bigvee}_{i=1}^k (y = y_i)\big) \big]\]
is an ${\cal L}$-step formula. By Lemma \ref{Lemma3}, we have $$\varphi_k(x,y_{1,k})\equiv \bigvee_{i_k \in I_k}a_{i_k} \wedge \chi_{i_k}(x,y_{1,k})$$ for some finite set $I_k$, semiring elements $a_{i_k}\in K$, and  formulas $\chi_{i_k}(x,y_{1,k})$ in ${\cal L}$. Then

\[
\begin{array}{ll}
& \exists y_{1,k}.\big[ \varphi_k(x,y_{1,k}) \wedge y_{1,k}\in X \wedge 
 \underbrace{\forall y. \big ((y\in X) \wedge \psi_{1}(x,y) \pimplies \underline{\bigvee}_{i=1}^k (y = y_i)\big)}_{=: \mu} \big] \\
 \equiv \;\; & \exists y_{1,k}.\big[ \big( \bigvee_{i_k \in I_k}a_{i_k} \wedge \chi_{i_k}\big) \wedge (y_{1,k}\in X) \wedge 
 \mu \big] \\[2mm]
 \equiv\;\; & \exists y_{1,k}.\big[ \bigvee_{i_k \in I_k} a_{i_k} \wedge \chi_{i_k}\wedge (y_{1,k}\in X) \wedge 
 \mu \big]\\[2mm]
 \equiv^\dagger\; &  \bigvee_{i_k \in I_k} a_{i_k} \wedge \underline{\exists}y_{1,k}.\big[\chi_{i_k}\wedge (y_{1,k}\in X) \wedge 
 \mu \big] 
\end{array}
\]
and the last formula is an ${\cal L}$-step formula because the formula
\[\underline{\exists}y_{1,k}.\big[\chi_{i_k}(x,y_{1,k})\wedge (y_{1,k}\in X) \wedge  
 \mu \big]\]
is in ${\cal L}$. Indeed, the first and the second conjunct are in  $\cal L$ and BFO, respectively, and $\mu$ is in BFO+mod. Moreover, we have assumed that $\bFOmod \subseteq {\cal L}\subseteq \bMSO$ and that $\cal L$ is closed under conjunction and $\underline{\exists}y_{1,k}$. 

In the step $\dagger$ of the reasoning we use the fact that, for every $i_k \in I_k$,  $\xi \in T_\Sigma$, $J\subseteq \pos(\xi)$, and $u\in\pos(\xi)$, there is at most one \underline{sequence} $u_{1,k}\in \pos(\xi)$ such that
\begin{equation*}
\seml\chi_{i_k}(x,y_{1,k})\wedge  (y_{1,k}\in X) \wedge 
 \mu \semr(\xi,J,u,u_{1,k})\neq 0.
\end{equation*}
The latter statement can be seen as follows. If the above inequality holds, then in particular $\seml \chi_{i_k}(x,y_{1,k}) \semr(\xi,u,u_{1,k})\neq 0$, which implies $\seml \varphi_{k,n}(x,y_{1,k}) \semr(\xi,u,u_{1,k})\neq 0$. 
Then also $\theta_k(X,x,y_{1,k})\semr (\xi,J,u,u_{1,k})\neq 0$ and we can apply Lemma \ref{technical-lemma}. Let us denote this sequence by $u_{1,k}^{(i_k)}$.

Having this uniqueness, we can prove $\dagger$ as follows:
\[
\begin{array}{cl}
  & \seml \exists y_{1,k}.\big[ \bigvee_{i_k \in I_k} a_{i_k} \wedge \chi_{i_k}\wedge (y_{1,k}\in X) \wedge \mu \big] \semr(\xi,J,u)\\[2mm]
= & \sum_{v_{1,k} \in \pos(\xi)} \sum_{i_k \in I_k} \seml a_{i_k} \wedge \chi_{i_k}\wedge (y_{1,k}\in X) \wedge \mu  \semr(\xi,J,u,u_{1,k})\\[2mm]
= & \sum_{v_{1,k} \in \pos(\xi)} \sum_{i_k \in I_k}  a_{i_k}  \cdot \seml \chi_{i_k}\wedge (y_{1,k}\in X) \wedge \mu  \semr(\xi,J,u,u_{1,k})\\[2mm]
= & \sum_{i_k \in I_k}  a_{i_k}  \cdot \seml \chi_{i_k}\wedge (y_{1,k}\in X) \wedge \mu  \semr(\xi,J,u,u_{1,k}^{(i_k)})\\[2mm]
= & \sum_{i_k \in I_k}  a_{i_k}  \cdot \seml \underline{\exists} y_{1,k}. \big[\chi_{i_k}\wedge (y_{1,k}\in X) \wedge \mu \big] \semr(\xi,J,u)\\[2mm]
= & \seml \bigvee_{i_k \in I_k}  a_{i_k}  \wedge \underline{\exists} y_{1,k}. \big[ \chi_{i_k}\wedge (y_{1,k}\in X) \wedge \mu \big] \semr(\xi,J,u)\enspace.
\end{array}
\]
\end{proof}

\subsection{Proof of Theorem \ref{EA-theo}}

Let $\xi \in T_\Sigma$ be an arbitrary tree throughout this section.

Let $\chi \in \uBTC_w^l(\Psi)$ for some $l\in \nat_+$ and $w\in \nat_+^*$, where $\Psi$ is the family of $n$-progress formulas introduced in Section \ref{bTC-section}. Recall that $\overline{\chi}$ has $l$ free variables of which the leftmost is $x_w$.  Then we define
$$\listpos(\xi,\overline{\chi})=\{u,u_{1,l-1} \in \pos(\xi)^l\mid \overline{\chi}^\xi(u,u_{1,l-1}) \text{ holds }\}.$$
Moreover, for every $w\in \nat_+^*$ and $\chi \in \uBTC_w^l(\Psi)$ we denote by
$\chi_\Phi$ the formula which we obtain by replacing every occurrence
of a subformula $\psi_{k}(x_v,x_{v1,k})$ of $\chi$ by
$\varphi_k(x_v,x_{v1,k})$. Note that 
$$\uBTC_w^{l}(\Phi)=\{ \chi_\Phi \mid \chi \in \uBTC_w^l(\Psi)\}.$$

Let $J \subseteq \pos(\xi)$ and $v \in \pos(\xi)$ be a base position. We define
\[\mathrm{\nextbase}_n(J,v) \text{ to be the vector } (v_1,\ldots,v_k) \in \pos(\xi)^k \]
uniquely determined by the following conditions:
\begin{itemize}
\item $k\geq 0$,
\item $\forall (1\le  i\le k): (v \le_n v_i)^\xi$ and $\exists (u_i \in J):  v_i = \langle u_i\rangle$,
\item $\forall (1\le  i< k): \sibl_n(v_i,v_{i+1})$, and
\item $\forall (v'\in \pos(\xi)):$   if   $(v \le_n v')$ and  $\exists (u' \in J):  v' = \langle u'\rangle$,  then $v'=v_i$  for some $1\le i \le k$.
\end{itemize}

Let $\mathrm{\nextbase}_n(J,v) = (v_1,\ldots,v_k)$ and let 
\begin{equation}\label{eq:Ji}
J_i = \{v \in J \mid (v_i \le v)^\xi\}
\end{equation}
 for every $1 \le i \le k$. It is easy to observe that the predicate $\sel^\xi$ is inductive on $J$ in the following sense.
\begin{ob}\rm\label{ob:BTC-number} Let $l\geq 1$,  $J\subseteq \pos(\xi)$ with $|J|=l$, and $v\in \pos(\xi)$ be a base position. 
Moreover, let $\mathrm{\nextbase}_n(J,v) = (v_1,\ldots,v_k)$ and let $J_i$ be defined as in (\ref{eq:Ji}) for every $1 \le i \le k$. Then
$\sel^\xi(J,v)$ holds, if and only if there is exactly one sequence $u,u_{1,k} \in J$ such that $v=\langle u\rangle$, $\psi^\xi_{k,n}(u,u_{1,k})$ and $v_i=\langle u_i\rangle$ and $\sel^\xi(J_i,v_i)$  for every $1 \le i \le k$.
\end{ob}
Moreover, let $w\in \nat_+^*$. Then we define the formula $f(J,v,w)$ inductively by
\[
f(J,v,w) = \exists x_{w1,k}.\psi_{k}(x_w,x_{w1,k})\wedge f(J_1,v_1,w1)\wedge  \ldots\wedge f(J_k,v_k,wk).
\]
We may call $f(J,v,w)$ the {\em $\psi$-formula determined by  $J$, $v$, and $w$}. Note that in general $f(J,v,w)$ is not an unfolding of $\Psi$. However, for  a set 
 $J\subseteq \pos(\xi)$ of positions with $\sel^\xi(J,v)$, the
formula $f(J,v,w)$ is an unfolding of $\Psi$ and $J$ can be considered as an assignment which satisfies $\overline{f(J,v,w)}$. We make this clear in the next lemma.

\begin{lm}\rm\label{lm:f-v-J} Let $l\geq 1$, $w\in \nat_+^*$, $J\subseteq \pos(\xi)$ with $|J|=l$, and $v\in \pos(\xi)$ be a base position. 
Then $\sel^\xi(J,v)$ holds, if and only if $f(J,v,w) \in \uBTC_w^l(\Psi)$ and there is exactly one enumeration $u,u_1,\ldots,u_{l-1}$ of $J$ such that $\langle u\rangle =v$ and $\overline{f(J,v,w)}^\xi(u,u_{1,l-1})$ holds.
\end{lm}
\begin{proof} By induction on $l$. Let $\mathrm{\nextbase}_n(J,v) = (v_1,\ldots,v_k)$, $J_i$ be defined as in (\ref{eq:Ji}), and $|J_i|=l_i$ for every $1 \le i \le k$. 

\underline{$l=1$:} Now $\mathrm{\nextbase}_n(J,v) = (\,)$ and $f(J,v,w)=\psi_{0}(x_w)$. Hence the statement trivially holds by the definition of $\sel(X,z)$.

\underline{$l\Rightarrow l+1$:} 

First we prove the implication $\Rightarrow$.  Since $\sel^\xi(J,v)$, by Observation \ref{ob:BTC-number}, there is a unique sequence $u,u_{1,k} \in J$ such that $v=\langle u\rangle$, $\psi^\xi_{k,n}(u,u_{1,k})$ and $v_i=\langle u_i\rangle$ and $\sel^\xi(J_i,v_i)$ for every $1 \le i \le k$. By the induction hypothesis,
for every $1 \le i \le k$, $f(J_i,v_i,wi) \in \uBTC^{l_i}_{wi}(\Psi)$ and there is a unique enumeration $u_i, p_i$ of $J_i$, such that $\overline{f(J_i,v_i,wi)}^\xi(u_i,p_i)$, where $p_i$ now denotes a sequence of length $l_i-1$. Since $l=l_1+\ldots+l_k$, we have $f(J,v,w) \in \uBTC^{l+1}_w(\Psi)$. Moreover,  for the enumeration $u,u_{1,k},p_1,\ldots,p_k$ of $J$, we have
$\overline{f(J,v,w)}^\xi(u,u_{1,k},p_1,\ldots,p_k)$.

Next we prove the implication $\Leftarrow$. Now $f(J,v,w) = \exists x_{w1,k}.\psi_{k}(x_w,x_{w1,k})\wedge f(J_1,v_1,w1)\wedge  \ldots\wedge f(J_k,v_k,wk)$, where $f(J_i,v_i,wi) \in \uBTC^{l_i}_{wi}(\Psi)$.
Then, we can decompose the given enumeration of $J$ into $u,u_{1,k},p_1,\ldots,p_k$ such that $\psi^\xi_{k,n}(u,u_{1,k})$ and 
$u_i,p_i$ is an enumeration of $J_i$ with $\overline{f(J_i,v_i,wi)}^\xi(u_i,p_i)$ for every $1\le i\le k$. By the induction hypothesis $\sel^\xi(J_i,v_i)$ for every $1 \le i \le k$.
Moreover, $v=\langle u\rangle$ and $v_i=\langle u_i\rangle$ for every $1\le i\le k$. By Observation \ref{ob:BTC-number}, this means that $\sel^\xi(J,v)$.
\end{proof}

In the following lemma we prove that, roughly speaking, for each unfolding of $\Psi$ and assigment satisfying it, the set of nodes appearing in the assigment determines the unfolding.

\begin{lm}\rm \label{lm:unf-sel} Let $v\in \pos(\xi)$ be a base position, $w\in \nat_+^*$, $\chi \in \uBTC_w^l(\Psi)$ and $u,u_{1,l-1} \in  \listpos(\xi,\overline{\chi})$ with $\langle u\rangle  = v$. Then $\sel^\xi(\{u,u_{1,l-1}\},v)$ holds and $\chi = f(\{u,u_{1,l-1}\},v,w)$.
\end{lm}
\begin{proof} We prove by induction on $l$.

\underline{$l=1$:} Then $\chi=\psi_{0}(x_w)$. By $\langle u\rangle  = v$ we have $\sel^\xi(\{u\},v))$. Moreover, $f(\{u\},v,w)=\psi_{0}(x_w)$.

\underline{$l\Rightarrow l+1$:} Now $\chi = \exists x_{w1,k}.\psi_{k}(x_w,x_{w1,k})\wedge \chi_1\wedge  \ldots\wedge \chi_k$, where $\chi_i \in \uBTC^{l_i}_{wi}(\Psi)$ for some $l_i$ such that $l=l_1+\ldots +l_k$. Moreover, $u,u_{1,l-1} = u,u_{1,k},p_1,\ldots,p_k$, where 
$\psi_{k}^\xi(u,u_{1,k})$, and $u_i,p_i \in \listpos(\xi,\overline{\chi_i})$ for every $1\le i \le k$. Then,
 by the induction hypothesis, we have $\sel^\xi(\{u_i,p_i\},v_i))$
with $v_i= \langle u_i \rangle $, and $\chi_i=f(\{u_i,p_i\},v_i,wi)$ for every $1\le i \le k$. Since $v=\langle u \rangle$ and 
$\psi_{1}^\xi(u,u_{i})$ for every $1\le i \le k$, we have $\sel^\xi(\{u,u_{1,k},p_1,\ldots,p_k\},v))$. Finally,
$\chi = \exists x_{w1,k}.\psi_{k}(x_w,x_{w1,k}) \wedge f(\{u_1,p_1\},v_1,w1)\wedge \ldots \wedge f(\{u_k,p_k\},v_k,wk)=f(\{u,u_{1,l-1}\},v,w)$.
\end{proof}

Now we are able to show that there is a bijection between sets $J\subseteq \pos(\xi)$ of positions with $\sel^\xi(J,v)$ and unfoldings of $\Psi$ with assignments which satisfy them.

\begin{lm}\rm \label{lm:bijection} Let $l\geq 1$, $w\in \nat_+^*$, and $v\in \pos(\xi)$ be a base position. There is bijection between the sets
\[G=\{(J,v) \mid J\subseteq \pos(\xi), | J| =l, \sel^\xi(J,v) \}\]
and
\[H_w=\{(\chi,u,u_{1,l-1})\mid \chi \in \uBTC_w^l(\Psi), u,u_{1,l-1} \in  \listpos(\xi,\overline{\chi}) \text{ with } \langle u\rangle =v\}.\]
\end{lm}
\begin{proof} Define the mapping
\[ G\to H_w \text{ by } (J,v) \mapsto \big(f(J,v,w),u,u_{1,l-1}\big),\]
where $u,u_{1,l-1}$ is the enumeration of $J$ appearing in Lemma \ref{lm:f-v-J}. It should be clear that this mapping is well-defined and injective. Moreover, by Lemma \ref{lm:unf-sel} it is surjective.
\end{proof}

\begin{lm}\rm \label{lm:unf-sel2} Let $w\in \nat_+^*$, $v\in \pos(\xi)$ be a base position and $J\subseteq \pos(\xi)$ such that $\sel^\xi(J,v)$.
For every $p\in J$, there is a unique integer $k(p)\geq 0$ and a sequence $p_{1,k(p)} \in J$ such that
$\psi_{k(p)}^\xi(p,p_{1,k(p)})$ holds, and 
$$\seml \overline{f(J,v,w)_\Phi}\semr(\xi,u,u_{1,l-1})= \prod\limits_{p\in J} \seml \varphi_{k(p)}(x,y_{1,k(p)})\semr (\xi,p,p_{1,k(p)} ),$$
where $u,u_{1,l-1}$ is the unique enumeration of $J$ appearing in Lemma \ref{lm:f-v-J}.
\end{lm}
\begin{proof} By induction on $|J|$.
\end{proof}

\

\begin{proof} of Theorem \ref{EA-theo}.  Let $\xi \in T_\Sigma$. Then we have:
\[
\seml  \BTC(\Phi)\semr (\xi,\varepsilon)
=  \sum\limits_{1 \le l \le \size(\xi)} \seml  \BTC_\varepsilon^l(\Phi)\semr (\xi,\varepsilon)
= \sum\limits_{1 \le l \le \size(\xi)} \,\sum\limits_{\chi \in \uBTC_\varepsilon^l(\Phi)} \,\seml \chi\semr(\xi,\varepsilon)
\]
where the last equation is due to  Lemma \ref{lm:alternative-semantics}. Next we change the index set of the second summation:
\[
\sum\limits_{\chi \in \uBTC_\varepsilon^l(\Phi)} \,\seml \chi\semr(\xi,\varepsilon)
= \sum\limits_{\chi \in \uBTC_\varepsilon^l(\Psi)} \,\seml \chi_\Phi\semr(\xi,\varepsilon)\enspace.
\]
Let $\chi \in \uBTC_\varepsilon^l(\Psi)$.  Since $\seml\overline{\chi_\Phi}\semr(\xi,u,u_{1,{l-1}}) \neq 0$  implies that  $\overline{\chi}^\xi(u,u_{1,l-1})$ holds for every $u,u_{1,{l-1}} \in \pos(\xi)$, we have:
\[
\seml \chi_\Phi\semr(\xi,\varepsilon)
= \sum\limits_{\substack{\varepsilon, u_{1,l-1} \in \pos(\xi)}}\,\seml \overline{\chi_\Phi}\semr(\xi,\varepsilon,u_{1,l-1})
= \sum\limits_{\substack{\varepsilon, u_{1,l-1} \in \listpos(\xi,\overline{\chi})}}\,\seml \overline{\chi_\Phi}\semr(\xi,\varepsilon,u_{1,l-1})\enspace.
\]

To summarize so far, we have 
\[
\seml  \BTC(\Phi)\semr (\xi,\varepsilon)
= \sum\limits_{1 \le l \le \size(\xi)}\, \sum\limits_{\chi \in \uBTC_\varepsilon^l(\Psi)} \,
\sum\limits_{\varepsilon,u_{1,l-1}\in \listpos(\xi,\overline{\chi})}\,
\seml \overline{\chi_\Phi}\semr(\xi,\varepsilon,u_{1,l-1}).
\]
Then we can continue as follows:
\[
\begin{array}{cl}
   & \sum\limits_{1 \le l \le \size(\xi)}\, \sum\limits_{\chi \in \uBTC_\varepsilon^l(\Psi)} \,
\sum\limits_{\varepsilon,u_{1,l-1}\in \listpos(\xi,\overline{\chi})}\,
\seml \overline{\chi_\Phi}\semr(\xi,\varepsilon,u_{1,l-1})\\[5mm]
=  & \sum\limits_{1 \le l \le \size(\xi)}\, 
\sum\limits_{\substack{J \subseteq \pos(\xi)\\|J|=l \\  \sel^\xi(J,\varepsilon)\\ \varepsilon \in J}}
\, \seml \overline{f(J,\varepsilon,\varepsilon)_\Phi}\semr(\xi,\varepsilon,u_{1,l-1}) \\[14mm]
& \text{(by Lemmas \ref{lm:unf-sel} and \ref{lm:bijection}, where }(J,\varepsilon) \mapsto (f(J,\varepsilon,\varepsilon),\varepsilon,u_{1,l-1} ))\\
& \text{ is the bijection in that lemma})\\[3mm]

=  & \sum\limits_{\substack{J \subseteq \pos(\xi)\\\sel^\xi(J,\varepsilon)\\ \varepsilon \in J}}\,
\seml \overline{f(J,\varepsilon,\varepsilon)_\Phi}\semr(\xi,\varepsilon,u_{1,|J|-1}) \\[12mm]
& \text{(where }(J,\varepsilon) \mapsto (f(J,\varepsilon,\varepsilon),\varepsilon,u_{1,|J|-1} ))\text{ is the bijection in Lemma \ref{lm:bijection}})\\[3mm]

= & \sum\limits_{\substack{J \subseteq \pos(\xi)\\\sel^\xi(J,\varepsilon)\\ \varepsilon \in J}}\, \prod\limits_{p\in J} \seml \varphi_{k(p)}(x,y_{1,k(p)})\semr (\xi,p,p_{1,k(p)} )\\[12mm]
& (\text{by Lemma \ref{lm:unf-sel2}})\\[2mm]

= & \sum\limits_{\substack{J \subseteq \pos(\xi)\\\sel^\xi(J,\varepsilon)\\ \varepsilon \in J}}\, \prod\limits_{p\in J} \seml \theta_{k(p)}(X,x,y_{1,k(p)})\semr (\xi,J,p,p_{1,k(p)} )\\[12mm]
& (\text{by the definition of } \theta_k \text{ and the fact that }k(w)\text{ is unique})\\[2mm]

= & \sum\limits_{\substack{J \subseteq \pos(\xi)\\\sel^\xi(J,\varepsilon)\\ \varepsilon \in J}}\, \prod\limits_{p\in J} \seml \bigvee_{k=0}^m\exists y_{1,k}. \;\theta_{k}(X,x,y_{1,k(p)})\semr (\xi,J,p )\\[12mm]
& (\text{by Lemma \ref{technical-lemma}})\\[2mm]

=  & \sum\limits_{\substack{J \subseteq \pos(\xi)\\\sel^\xi(J,\varepsilon)\\ \varepsilon \in J}}\,
\prod\limits_{p \in \pos(\xi)}\seml (x\in X) \pimplies
\bigvee_{k=0}^m\exists y_{1,k}. \; \theta_k(X,x,y_{1,k}) \semr (\xi,J,p)\\[10mm]
=  & \sum\limits_{J \subseteq \pos(\xi)} \prod\limits_{p \in \pos(\xi)}\seml \theta(X,x)\semr(\xi,J,p) \\[5mm]
=  & \seml \exists X. \forall x. \theta(X,x)\semr(\xi)\\[2mm]
=  &  \seml\Theta\semr(\xi) 
\end{array}
\]
\end{proof}

\section{From RMSO-definability to Recognizability}

Here we show that every RMSO-definable weighted tree language is recognizable. We prove this as usual by induction on the structure of the formulas.

\begin{theo}\label{th:rec=definable} Let  $r: T_\Sigma \rightarrow S$ be a weighted tree language. If $r$ is  $\RMSO$-definable, then $r$ is recognizable.
\end{theo}
\begin{proof} Let $\varphi$ be an RMSO-formula.  If $\varphi$ has the form $a$, $\lab_\sigma(x)$,
$\edge_i(x,y)$,  $x \in X$, $\varphi \wedge \psi$,
$\varphi \vee \psi$, $\exists x. \varphi$, or $\exists X. \varphi$,
then  we can proceed as in \cite[Lm. 5.2-5.4]{drovog06} showing that $\seml \varphi \semr$ is recognizable. 

For the formula $x \le y$ we apply Observation \ref{ob:BMSO-rec}. 

Next we consider a formula $\varphi$  of the form   $\neg \psi$ where $\psi$ is a $\bMSO$-step formula. Then $\varphi$ is also a $\bMSO$-step formula, and by Lemma \ref{Lemma3}   its semantics is a recognizable step function. Using the fact that recognizable weighted tree languages are closed under scalar product and summation (cf.  \cite[Lm. 3.3 of]{drovog06}), we obtain that the semantics of $\varphi$ is recognizable.  

Next let $\varphi$ be of the form $\forall x. \psi$ where $\psi$ is a  BMSO-step formula. By Lemma \ref{Lemma3}, the semantics of $\psi$ is a recognizable  step function and thus by  \cite[Lm. 5.5]{drovog06} the semantics of  $\varphi$ is  recognizable.

Finally let $\varphi$ be of the form $\forall X. \chi$ where $\chi$ is a $\bMSO$-formula. Then also $\varphi$ is a $\bMSO$-formula of which the semantics is a recognizable weighted tree language again due to Observation \ref{ob:BMSO-rec}. 

Thus, summing up, every  $\RMSO$-definable weighted tree language is recognizable. 
\end{proof}

In the present paper we have defined the fragment  $\RMSO$ of restricted MSO in the spirit of
\cite{gas10} (and of \cite{bolgasmonzei10}). It is syntactically slightly different from the fragment with the same name introduced in \cite{drogas05,drogas07} and used in \cite{drovog06,drovog11} for the tree case. In the restricted MSO-fragment of \cite{drogas05,drogas07}, (cf. e.g.  \cite[Def. 4.1 and 4.8]{drovog06}) $x\le y$ is not an atomic formula, negation is only applicable to atomic formulas
except coefficients from $S$, and  second-order universal
quantification is not allowed. Henceforth we will call this fragment
$\RMSO'$ (cf.  \cite[Def. 4.8]{drovog06}). In  \cite[Thm. 5.1]{drovog06} is was proved that a weighted tree language (over a commutative semiring) is $\RMSO'$-definable if, and only if it is recognizable. Due to Theorem \ref{main} we obtain the following corollary.

\begin{corollary} Let $r$ be a weighted tree language. Then, $r$ is $\RMSO$-definable if, and only if $r$ is $\RMSO'$-definable.
\end{corollary}

{\bf Acknowledgement.} The authors are grateful to one of the referees  for his/her
careful analysis and helpful remarks which definitely improved the quality of the paper.

\bibliographystyle{alpha}

\appendix

\section{Collection of the Used Macros}

For the convenience of the reader we list here all the macros which are used in this paper.

For every $\varphi,\psi \in \MSO$:
\begin{compactitem}
\item $\varphi \stackrel{+}{\rightarrow} \psi:= \neg \varphi \vee (\varphi \wedge
\psi)$ \hspace{5mm} (cf. p.\pageref{p:+})
\end{compactitem}

\noindent For every $\varphi,\psi \in \bMSO$ (cf. p.\pageref{p:B}):
\begin{compactitem}
\item $\varphi \underline{\vee} \psi := \neg(\neg \varphi \wedge \neg \psi)$ 
\item $\underline{\exists} x. \varphi := \neg \forall x. \neg \varphi$
\item $\underline{\exists} X. \varphi := \neg \forall X. \neg \varphi$
\end{compactitem}

\noindent Next we proceed from the simpler macros to the more complex ones.
\begin{compactitem}
\item $\edge(x,y) := \underline{\bigvee}_{1\leq i \leq \maxrk(\Sigma)} \edge_i(x,y)$
\hspace{5mm} (cf. p.\pageref{p:le_w})

\

\item $\rt(x) :=  \forall y. \neg\,\edge(y,x)$ \hspace{5mm} (cf. p.\pageref{p:le_w})

\

\item $\formpath_w(y_{0,n}) := \bigwedge\limits_{1 \le i \le n}
\edge_{w_i}(y_{i-1},y_i)$ \hspace{5mm} (cf. p. \pageref{p:le_w})

\

\item $(x \le_w y) := \underline{\exists} y_{0,n}. (x=y_0) \wedge
\formpath_w(y_{0,n}) \wedge (y_n=y)$ \hspace{5mm} (cf. p. \pageref{p:le_w})

\

\item $x \le_n y := \underline{\bigvee}_{w \in \{1,\ldots,\maxrk(\Sigma)\}^n} (x \le_w y)$ \hspace{5mm} (cf. p. \pageref{p:le_w})

\

\item $(\height(x) \ge n) := \underline{\exists} z. (x \le_nz)$ \hspace{5mm} (cf. p. \pageref{p:height})

\

\item $\sibl(x, y) := \underline{\exists} z. \underline{\bigvee}_{1\leq i < j \leq \maxrk(\Sigma)} \edge_i(z,x)\wedge \edge_j(z,y)$  \hspace{5mm} (cf. p. \pageref{p:BTC})

\

\item $\sibl_n(x,y) := 
\underline{\exists} x',y'. \big( \sibl(x',y') \wedge (x'\le_{n-1} x) \wedge
(y' \le_{n-1} y)\big)$  \hspace{5mm} (cf. p. \pageref{p:BTC})

\

\item $(|y| > n) := \underline{\exists} x. \; (x \le_{n+1} y)$ \hspace{5mm} (cf. p. \pageref{p:le_w})

\

\item $(y/n \in X) := \underline{\exists} x. \; (x\in X) \wedge (x \le_n y)$ \hspace{5mm} (cf. p. \pageref{p:le_w})

\

\item $(m \in |X|) := \underline{\exists} x,y. \; \rt(x) \wedge (x \le_m y) \wedge (y \in X)$ \hspace{5mm} (cf. p. \pageref{p:le_w})

\

\item $(|x| \equiv_n m) := $\\
$ \forall X. \left(\left( (x\in X) \wedge \left(\forall y. ((y\in X) \wedge (|y| > n)) \stackrel{+}\rightarrow (y/n \in X)\right)\right)\stackrel{+}\rightarrow (m \in |X|)\right)$ \\
(cf. p. \pageref{p:|x|equiv_n})

\

\item $(y=\langle x\rangle_n) := \bigwedge_{0\leq q < n} \big((|x|
  \equiv_n q) \stackrel{+}\rightarrow (y \le_q x)\big)$ \hspace{5mm} (cf. p. \pageref{p:BTC})

\

\item $\formpath(y_{0,n}) := \underline{\bigvee}_{w \in \{1,\ldots,\maxrk(\Sigma)\}^n}
\formpath_w(y_{0,n})$ \hspace{5mm} (cf. p. \pageref{p:onlmp})

\

\item $\formlmp(y_{1,n}) :=
  \formpath(y_{1,n}) \wedge$ \\
\hspace*{27mm} $ \forall z.\bigg[\bigg((y_1 \le_{n-1} z) \wedge (z \not= y_n)\bigg)
\stackrel{+}{\rightarrow}\underline{\bigvee}_{i=1}^{n-1} \; \sibl_{n-i}(y_{i},z') \bigg]$ \\
(cf. p. \pageref{p:onlmp})

\

\item $\onlmp_{n-1}(x,y) := \underline{\exists} y_{1,n}. (x=y_1) \wedge
  \formlmp(y_{1,n}) \wedge \underline{\bigvee}_{1 \le i
    \le n} (y_i=y)$ \\
(cf. p. \pageref{p:onlmp})

\

\item $\formcut_{n,k}(x,y_{1,k}):=  \left(\bigwedge_{i=1}^{k} (x \le_n y_i)\wedge (\height(y_i) \geq n)\right) \wedge$\\ 
$ \left(\bigwedge_{i=1}^{k-1} \sibl_n(y_i, y_{i+1})\right) \wedge
\left(\forall z. ((x \le_n z) \wedge (\height(z) \geq n))  \stackrel{+}{\rightarrow}\big(\underline{\bigvee}_{i=1}^k z=y_i\big)\right)$\\
(cf. p. \pageref{p:onlmp})

\

\item $\chk_\zeta(x,y_{1,k}) := \bigwedge\limits_{w \in
  \pos(\zeta) \setminus \pos_{Z_k}(\zeta)} (\underline{\exists}
y. \; (x \le_w y) \wedge \mathrm{label}_{\zeta(w)}(y))
\wedge$\\
\hspace*{30mm}$\bigwedge\limits_{1 \le i \le k} (x\le_{\pos_{z_i}(\zeta)} y_i)$
\hspace{5mm} (cf. p. \pageref{p:check})

\end{compactitem}

\end{document}